\documentclass[11pt]{article}

\usepackage{booktabs} 

\usepackage{fullpage}
\usepackage{mathtools}
\usepackage{amsmath,amsthm, amssymb}
\usepackage{latexsym}
\usepackage{graphicx}
\usepackage{ifthen}
\usepackage{subcaption}
\usepackage{multicol}
\usepackage{algorithm,algorithmic}
\usepackage{dsfont}
\usepackage{color}
\usepackage{xspace}
\usepackage{microtype} 
\usepackage{nicefrac}  
\usepackage{enumitem}

\usepackage{hyperref}
\hypersetup{colorlinks=true,linkcolor=blue,citecolor=blue,urlcolor=blue}

\newtheorem{lemma}{Lemma}
\newtheorem{theorem}{Theorem}
\newtheorem{definition}{Definition}
\newtheorem{corollary}{Corollary}

\newtheorem{remark}{Remark}

\usepackage{thm-restate}

\newcommand{\shouldWeRemove}[1]{}
\newcommand{\sacomment}[1]{}     
\newcommand{\sacommentMinor}[1]{}

\newcommand{\todo}[1]{}

\newcommand{\removeEC}[1]{}
\newcommand{\delete}[1]{}

\DeclareMathOperator{\Rev}{\mathrm{Rev}}
\newcommand{\rev}{\mathrm{Rev}}
\DeclareMathOperator{\Mye}{\mathrm{Mye}}
\newcommand{\revmye}{\rev^{\mathrm{Mye}}}

\newcommand{\eplmax}{\epl_{\text{max}}}

\DeclareMathOperator{\E}{\mathbb{E}}

\DeclareMathOperator{\argmax}{argmax}

\newcommand{\alloc}{x}
\newcommand{\payment}{p}

\newcommand{\threshold}{H}
\newcommand{\qu}{q}
\newcommand{\qum}{q_{m_b}}
\newcommand{\quplus}{q^{\dagger}}
\newcommand{\quan}{\qu_n}
\newcommand{\quanplus}{\quplus_n}
\newcommand{\quansplus}{\quplus_{\nsoph}}
\newcommand{\myeprice}{p_{m_b}}

\newcommand{\ep}{\ell}
\newcommand{\prob}{\theta}
\newcommand{\badrounds}{R}

\newcommand{\resg}{r^g}
\newcommand{\resb}{r^b}
\newcommand{\epl}{\mathcal E}
\newcommand{\Ex}{\mathbb{E}}

\newcommand{\Kbound}{\frac{10}{\epsilon(1-\epsilon)}}
\newcommand{\KboundE}{\frac{8}{\epsilon(1-\epsilon)}}
\newcommand{\rhobound}{\frac{\epsilon(1-\epsilon)(1-\delta)(1-\rho)}{12(1+\epsilon)}}
\newcommand{\rhoCompactBound}{\frac{\epsilon(1-\epsilon)^4}{12}}
\newcommand{\deltaH}{\sqrt{\frac{4\log(1/\epsilon)}{\threshold}}}
\newcommand{\Hdelta}{\frac{4\log(1/\epsilon)}{\delta^2}}
\newcommand{\eplVal}{\frac{2\threshold m_g}{(1-\delta)(1 - \rho)}}
\newcommand{\eplmaxVal}{\frac{2\threshold n}{(1-\delta)(1 - \rho)}}
\newcommand{\uncleared}{U}
\newcommand{\allocated}{A}
\newcommand{\val}{v}

\newcommand{\valit}{\val_{it}}
\newcommand{\Val}{V}

\newcommand{\bid}{b}

\newcommand{\bidi}{\bid_{i}}
\newcommand{\bidj}{\bid_{j}}
\newcommand{\bidit}{\bid_{it}}

\newcommand{\nsoph}{n_{\text{soph}}}
\newcommand{\nnaive}{n_{\text{naive}}}
\newcommand{\nnaiveg}{n_{\text{naive,g}}}
\newcommand{\naive}{\text{naive}}
\newcommand{\ksoph}{k_{\text{soph}}}

\newcommand{\goodauction}{\textsc{Good-State-Auctions}\xspace}
\newcommand{\badauction}{\textsc{Bad-State-Auctions}\xspace}

\newcommand{\INPUT}{\item[{\bf Input:}]}
\renewcommand{\algorithmiccomment}[1]{\bgroup\hfill\footnotesize~#1\egroup}

\begin{document}

\title{Dynamic First Price Auctions Robust to Heterogeneous Buyers}
\author{Shipra Agrawal\footnote{Columbia University. sa3305@columbia.edu. This research was supported in part by a Google Faculty Research Award and NSF CAREER award CMMI-1846792.} 
\and  Eric Balkanski\footnote{Harvard University. ericbalkanski@g.harvard.edu.  This research was supported in part by a Google PhD Fellowship.}  
\and  Vahab Mirrokni\footnote{Google Research, New York. mirrokni@google.com.} 
\and Balasubramanian Sivan \footnote{Google Research, New York. balusivan@google.com.} 
}
\date{}

\maketitle

\begin{abstract}
We study dynamic mechanisms for optimizing revenue in repeated auctions, that are robust to heterogeneous forward-looking and learning behavior of the buyers. Typically it is assumed that the buyers are either all  myopic or  are all infinite lookahead, and that buyers understand and trust the mechanism. These assumptions raise the following question: is it possible to design approximately revenue optimal mechanisms when the buyer pool is heterogeneous? Facing a heterogeneous population of buyers with an unknown mixture of $k$-lookahead buyers, myopic buyers, no-regret-learners and no-policy-regret learners, we design a simple state-based mechanism that achieves a constant fraction of the optimal achievable revenue. 
\end{abstract}


\newpage

\section{Introduction}
\label{sec:intro}

Bundling increases revenue in commerce: this is a well known fact widely used in practice (e.g., Amazon often suggests a ``frequently bought together'' bundle while browsing many items). This same principle is at the heart of why stateful repeated/dynamic auctions are significantly more lucrative than one-shot auctions: stateful dynamic auctions profitably use the opportunity to bundle across time. 
This revenue opportunity has inspired a series of recent works~\cite{PPPR16,ADH16,MPTZ16a,BMP16,LiuP17,MLTZ18,ADMS18,BMSW18,BMLZ19} on designing dynamic auctions satisfying various desired properties.

The ability to link auctions across time significantly expands the design space of auctions allowing some exceedingly complex auctions. How does a buyer optimize when bidding in such auctions? The standard and widely used notion in the literature of incentive compatibility (IC) requires that the buyers understand these complicated auctions well, and have an ``infinite lookahead''. I.e., they require the buyer to think about the consequences of their bid in the current round on all future round utilities and optimize current round bid accordingly. In particular, as~\cite{ADMS18} note, this assumes that the buyers understand the mechanism deeply enough to optimally respond, they believe that their interactions with the seller will last for all future rounds accounted for when computing utility, believe that the seller is indeed strictly following the advertised  mechanism etc.

Numerous practical reasons make these assumptions far from true in reality:  buyer's computational limitations, inability to predict the future well enough, inability to trust a seller or understand the exact auction that is run in a complex supply chain of auctions. A striking example of this is the display ads market in Internet advertising. Given the number of ad exchanges, and the variety of purchase mechanisms that are constantly evolving over time, buyers are often unable to trust or verify whether a seller has stuck to an announced mechanism. 

The consequence is that the seller faces a heterogeneous buyer population employing a large spectrum of strategies to maximize their perceived utility. As~\cite{ADMS18} note, such a buyer could behave myopically, or have a limited lookahead (i.e., a $k$-lookahead instead of infinite lookahead), or be a learner that makes decisions only based on past performance of various bidding strategies thereby completely disregarding the seller's description/promises in the mechanism's future.

The solution to tackling a heterogeneous buyer behavior cannot be a buyer-specific auction that tailors the optimal auction for a given buyer behavior. Implementing such a discriminative auction may be legally infeasible and also impractical (buyer's behavior may even change over time). Motivated by these observations,~\cite{ADMS18} 
 considered the setting of a single seller repeatedly interacting with a {\it single} buyer whose behavior (myopic/infinite lookahead/learner etc.) the seller is a priori unaware of, and designed a single mechanism that simultaneously obtains a constant fraction of the optimal revenue achievable against each potential buyer type. In this paper, we  study the following question:
\begin{center}
\emph{Is there an $n$-buyer mechanism that is robust against heterogeneous buyer behaviors?}
\end{center}
We consider a general setting with an arbitrary and unknown mixture of \emph{multiple} heterogeneous buyers. The challenges introduced by considering multiple buyers are discussed after describing the setting and main result.

\paragraph{The setting.}
We study a repeated interaction between a single seller and $n$ buyers over $T$ rounds. At the 
beginning of 
each round $t = 1,2,\ldots, T$, there is a 
single fresh good for sale whose private value $\val_{i,t} \in \Val$ 
for buyer $i$ is drawn, independently from other $j \neq i$ and $t' \neq t$, from a publicly known distribution
$F$ with finite expectation $\mu$. The buyer observes the valuation $\valit$ before making a bid $\bidit$. The good for sale in round $t$ has to be either allocated to one of the buyers or discarded immediately. Each buyer's valuation is additive across rounds. We consider a range of buyer behaviors similar to~\cite{ADMS18}, including myopic buyers, $k$-lookahead buyers, no-regret learners and no-policy-regret learners. Formal definitions of these different buyer behaviors are provided in Section \ref{sec:prelim}.

We  categorize as {\it sophisticated} buyers the buyers who are either $k$-lookahead for $k \geq \ksoph = \Theta(n)$ or  no-policy-regret learners, and {\it $\naive$} buyers  those who are myopic ($0$-lookahead) or no-regret learners. Given a population with $\nsoph$ sophisticated and $\nnaive=n-\nsoph$ $\naive$ buyers, a robust mechanism aims to achieve close to maximum per-round revenue achievable from such a population, without prior knowledge of which buyer is $\naive$ or sophisticated, or the values of $\nsoph$ and  $\nnaive$.

\paragraph{Upper bound on revenue.} With only  $\nnaive$ \naive \ buyers, e.g. $\nnaive$ myopic buyers, it is impossible in each round to get more than the optimal  $\revmye(\nnaive)$  revenue in a single-round auction obtainable from Myerson's auction~\cite{M81}.
With only $\nsoph$ sophisticated buyers, it is impossible to get more than $\E_{v_1,\dots,v_{\nsoph}\sim F}\left[\max_{i }  v_i\right]$ revenue, as revenue is upper bounded by the maximum value. It is easy to show that $\E_{v_1,\dots,v_{\nsoph}\sim F}\left[\max_{i } v_i\right] \leq \quansplus :=  \E_{v\sim F}\left[v | v \geq F^{-1}(1-1/\nsoph)\right]$ (see Appendix~\ref{sec:rev_up}). In  a setting with $\nsoph$   \emph{and} $\nnaive$  buyers, the total revenue achievable is thus at most $\quansplus + \revmye(\nnaive)$ (see Appendix~\ref{sec:rev_up}). Motivated by this upper bound, we define the following benchmark.

\begin{definition}
We call a mechanism {\bf $(\alpha, \beta)$-robust} if, for any per-round 
valuation distribution $F$, and for every value of $\nsoph \in [n]$ (and $\nnaive = n - \nsoph$), without knowing $\nsoph$ or $\nnaive$, it achieves an expected per round revenue of at least  $\alpha\cdot\revmye(\nnaive) + \beta\cdot\quansplus - o(1).$
\end{definition}

\paragraph{Main result.} We construct a mechanism that is Interim Individually Rational (IIR)\footnote{\cite{ADMS18} show that it is impossible to achieve such a high-revenue with per-round ex-post IR unless the buyer's lookahead is very high even in a single buyer setting, which rules that out as well in our setting.} 
and $(\Theta(1), \Theta(1))$-robust. I.e., for every value of $\nsoph,\nnaive$, without knowing $\nsoph$ or $\nnaive$, it achieves a per-round revenue within a constant factor of the near optimal
$\revmye(\nnaive) + \quansplus-o(1)$ revenue. The mechanism is a simple-to-implement first price auction  with  reserve price  based on the current state of buyers (where state is a succinct summary of the buyer's history). With the display ads industry moving to use first-price auctions~\cite{adageGoogle,fpera,rubicon}, this result is especially significant.

\subsection{Overview of challenges and technical approach}  The interactions between $n$ buyers, with heterogeneous and unknown lookahead  and learning behaviors, introduce multiple challenges compared to the single buyer setting studied in \cite{ADMS18}.

 
 Firstly an equilibrium, which is a profile of mutually best responding strategies from all agents, and a widely used predictor of a mechanism's outcome, is unlikely to exist in our setting. Therefore it is not possible to prove revenue guarantees by arguing about what revenue would be obtained in the equilibrium outcome. Indeed, far from being able to pinpoint what our mechanism's outcome will be, we only guarantee what the mechanism's outcome \emph{will not be}, as long as all agents satisfy \emph{much weaker} notions of rationality like playing undominated strategies~\cite{BLP06} or no-regret strategies. This aligns well with our assumptions of heterogeneity and distrust among buyers. With just this guarantee of which outcomes will not occur in our mechanism, we are able to establish our strong revenue guarantees. In contrast, in the 1-buyer setting of~\cite{ADMS18}, there are no equilibrium concerns: the single buyer simply best responds according to his utility function.
 
 Another non-trivial challenge is that while proving results in undominated strategies, we have to establish a given strategy for buyer $i$ being dominated regardless of another buyer $j$'s strategies. In particular $j$'s strategy could be an arbitrary function of the entire history of not only $j$'s own bids and outcomes, but also those of $i$. This creates complex dependencies of buyer $i$'s future utility on his bid today. For example, even if $i$ bids truthfully according to his valuations today, he has to consider a strategy of $j$ that would bid very high and sabotage $i$ in the future if it sees $i$ bidding beyond a certain bid in the current round. It requires careful design choices in the mechanism to be able to guarantee certain buyer behavior under such complex side-effects of a buyer's bid. Key aspects of our mechanism are designed to allow lower bounding the utility of a buyer irrespective of other buyers' behavior. This ability is crucial for establishing our revenue guarantees.
 

We construct a state-based mechanism where   sophisticated buyers  remain in a good state yielding high revenue, and where $\naive$ buyers  bid  as in an approximately revenue optimal one round auction. The mechanism incentivizes sophisticated buyers to remain in this good state \emph{in spite of the arbitrariness of other buyers' responses}, while also achieving high revenue.  To guarantee high utility to every sophisticated buyer in the good state, the mechanism temporarily ``rests'' buyers who have recently been allocated a large  number of items. Ignoring such buyers for a small number of rounds guarantees some rounds for each buyer to enjoy positive utility and get the item as long as they bid high enough. To do well against the benchmark of $\quansplus + \revmye(\nnaive)$ that extracts the optimal revenue possible from every buyer category, the mechanism necessarily has to  adapt to the mixture  of buyer population it is observing. Incorporating this adaptivity, and handling the resulting complications in utility analysis are some further challenges that we handle in our mechanism design and analysis. 


\subsection{Related work}

There are several streams of literature related to our work in dynamic mechanism design. 

\paragraph{Revenue maximization in dynamic auctions.} The stream closest to our paper is the work on revenue maximization in repeated auctions \cite{PPPR16,ADH16,MPTZ16a,BMP16,LiuP17,MLTZ18,ADMS18,BMSW18,BMLZ19}. The main difference of our work from all of these, with the exception of ~\cite{ADMS18, BMSW18}, is the notion of dynamic incentive compatibility which assumes that all buyers are infinite lookahead, while we allow arbitrary mixtures of buyer attitudes. Agrawal et al.~\cite{ADMS18} study the $1$-buyer setting and design robust auctions when the buyers are $k$-lookahead buyers or no-simple-regret learners or no-policy-regret learners. Braverman et al.~\cite{BMSW18} also study the $1$-buyer setting and design mechanisms to extract more than Myerson's revenue when the buyers follow no-simple-regret learning strategies, and in particular a subset of them called mean-based bidding strategies.

\paragraph{Repeated interactions with evolving values.} There is a large body of work on designing mechanisms in repeated interactions when the buyers' values evolve over time. See~\cite{BB84,Besanko85,CL00, Battaglini05, ES07,AS13,KLN13,PST14,BS15,CDKS16} for an overview of dynamic mechanisms in such settings.

\paragraph{Bargaining, durable goods monopolist and Coase conjecture.} Unlike our setting where the value is drawn independently in every round, there 
is a large body of literature in economics that studies settings where the value 
is initially drawn from a distribution but remains fixed later on. This setting 
can be motivated based on several applications including bargaining, durable 
goods monopoly and behavior based discrimination. See~\cite{FV06} for an 
excellent survey and references therein for an overview of this area and~\cite{DPS15,ILPT17} for work in the theoretical computer science literature. 

\paragraph{Lookahead search.}
	The study of $k$-lookahead search can be viewed in the context 
	of {\em bounded rationality}, as pioneered by Herb Simon \cite{Sim55}.  He argued that,
	instead of optimizing, agents may apply a class of heuristics, termed satisfying heuristics in
	decision making. A natural choice of such heuristics is restricting
	the search space of best-response moves.
	Lookahead search in decision-making has been motivated and examined in great extent by the artificial intelligence community~\cite{N83,KRS92,SKN09}. Lookahead search is also related to the sequential thinking framework in game theory \cite{SW94}.  More recently, ~\cite{MTV12} study the quality of equilibrium outcomes for look-ahead search strategies for various classes of games.
	They observe that the quality of resulting equilibria increases in generalized second-price auctions, and duopoly games, but not in other classes of games. No prior work studies dynamic mechanisms that are robust against various lookahead search strategies.
	

\section{Preliminaries}
\label{sec:prelim}
There are $n$ buyers and a single item for auction at each of $T$ sequential time steps. The value of buyer $i \in [n]$ at time step $t \in [T]$ for the item is denoted by $v_{i,t}$ and is drawn i.i.d.  from a common prior distribution $F$. The distribution $F$ is known to all the buyers and the seller. The realization of the value $v_{i,t} \sim F$  is private and is visible only to buyer $i$. In each round $t$, every buyer makes a bid $b_{i,t}$. A buyer may use history of bids, allocation, and payments before round $t$ along with its own private valuations until round $t$, to decide the bid in round $t$. The seller uses the entire bid profile ${\bf b}_t = (b_{1,t}, \ldots, b_{n,t})$ at time $t$, along with history (bids, allocations, and payments) before time $t$ to decide allocation $x_{i,t} \in \{0,1\}$ and payment $p_{i,t}\ge 0$ for each buyer, such that $\sum_{i=1}^n x_{i,t}\le 1$.
 
More precisely, the bidding, payment, and allocation strategy space are defined as follows.
\begin{definition}
Let $H_{t-1}:=({\bf b}_1, {\bf x}_{1}, {\bf p}_{1},\ldots, {\bf b}_{t-1}, {\bf x}_{t-1}, {\bf p}_{t-1})$ denote the history of bids, allocation, and payment  before time $t$. 
\end{definition}

\begin{definition}[Bidding strategy]
\label{def:strategy}
Buyer $i$'s bids in round $t$ are decided by functions $b_{i,t}=s_{i}(H_{t-1}, v_{i,t})$, i.e. bids in round $t$ are (possibly randomized) functions of history $H_{t-1}$ and buyer's private valuation in round $t$.
\end{definition}
\begin{definition}[Payment and allocation function]
For each buyer $i$, payment and allocation in round $t$ are given by $x_{i,t} = \alloc_i(H_{t-1}, {\bf b}_t)$ and $p_{i,t}=\payment_i(H_{t-1}, {\bf b}_t)$ of history $H_{t-1}$ and bid profile in round $t$. Here, $\alloc_i(H_{t-1}, {\bf b}_t) \in [0,1]$  denotes the probability of allocation to bidder $i$, and $\payment_i(H_{t-1}, {\bf b}_t)\ge 0$ denotes the bidder's expected payment.
\end{definition}
\begin{definition}[Buyer's realized utility]
\label{def:utility} 
Buyer $i$'s realized utility in a round $t$ is given by $u_{i,t}=v_{i,t} \alloc_{i,t} -\payment_{i,t}$, i.e., it is the difference between valuation and payment if the good is allocated to the buyer, and $0$ otherwise. We often refer to this as simply the buyer's utility in a round.
\end{definition}

\begin{remark}
Our current definition of history includes everything from past that  can possibly be revealed to buyers and can be relevant to a buyer's bidding strategy. However, our results are not tied to this particular definition of history. For example, instead of revealing the entire vector of bids from the past, if we happen to just reveal the allocation and payment after each auction, that would simply restrict the buyer's strategy space further. Our arguments will hold as long as a buyer can see their own bids, allocation, and payments in the past, and whether or not the good was allocated in each round.
\end{remark}

\begin{remark}
Technically, the history for buyer $i$ should also include their own private valuations in the past. However, since the valuations are generated independently in each round from a known valuation distribution, are visible only to the buyer, and the buyer's total utility is additive across rounds, the past valuations are irrelevant for the buyer's bidding strategy. Therefore, for simplicity, we chose to eliminate them from the definition of buyer's history.  
\end{remark}
Also define buyer i's history, $H_{i,t-1}:= (H_{t-1}, v_{i,1}, \ldots, v_{i,t-1})$.
would things should be still fine: this is because whatever be the strategies other buyers use, in a domination argument, we are going to fix those strategies, and then show that a chosen strategy for a protagonist buyer is dominated by another strategy.

In this paper, we focus on the first price auction mechanism which is defined by the following specific allocation and payment function, with flexibility to choose the mechanism for setting a ``reserve price".

\begin{definition}[First price auction with reserve price]
In a first price auction of a single good with reserve price $r$, the seller observes the bids of participating buyers, and then allocates the good to the buyer with highest bid if that bid is above or equal to the reserve price. The winning buyer's payment is equal to their bid.
\end{definition}

\subsection{Heterogeneous lookahead behavior}
We define heterogeneous forward-looking behavior of buyers by considering buyers who may be myopic or $k$-lookahead for different values of $k$. A myopic ($k$-lookahead) buyer is defined as a buyer who optimizes her myopic ($k$-lookahead) utility in every round to decide the bid. Below, we give precise definitions of these. Intuitively, myopic buyers optimize their current round utility, while $k$-lookahead buyers ($k\ge 1$) optimize their total expected utility over the current and next $k$ rounds.

\begin{definition}[Buyer's myopic utility]
Under bidding strategies ${\bf s} = \{s_j(\cdot)\}_{j=1,\ldots, n}$, the myopic utility of buyer $i$ in round $t$, given private valuation $v_{i,t}$ and history $H_{t-1}$, is defined as
$$U^{[t,t]}_i( H_{t-1}, v_{i,t}, {\bf s}) = \E_{v_{j,t}\sim F, j\ne i} \left[v_{i,t} \cdot \alloc_{i}(H_{t-1}, {\bf b}_t) -  \payment_{i}(H_{t-1}, {\bf b}_t) ; \  b_{j,t}=s_j(H_{t-1}, v_{j,t}), \forall j\right]$$
\end{definition}
\begin{definition}[$k$-lookahead utility]
Under bidding strategies ${\bf s} = \{s_j(\cdot)\}_{j=1,\ldots, n}$, the $k$-lookahead utility of buyer $i$ in round $t$, given private valuation $v_{i,t}$ and history $H_{t-1}$, is defined as
\begin{align*}
U^{[t,t+k]}_i\bigg( H_{t-1}, v_{i,t}, {\bf s}\bigg)
= \E_{v_{j,t}, j\ne i, {\bf v}^{[t+1, t+k]}}\sum_{r=0}^k \E_W\left[U^{[t+r,t+r]}_i\bigg( (H_{t-1}, W_{[t,t+r-1]}), v_{i,t+r}, {\bf s}\bigg)\right] 
\end{align*}

where $W_{t+r}$ denotes the vector of realized bids, allocation and payments in round $t+r$, $W_{[t,t+r-1]}:=\{W_t, \ldots, W_{t+r-1}\}$. And ${\bf v}^{[t+1, t+k]}$ denotes the realizations of private valuations for all bidders from time $t+1$ to $t+k$, i.e., ${\bf v}^{[t+1, t+k]}=\{v_{\tau,j}, \tau=t+1,\ldots, t+k,j=1,\ldots, n\}$. 

\end{definition}

\paragraph{Undominated strategies.}
\label{sec:undominated}
We provide revenue guarantees for our mechanism under an assumption that all myopic or $k$-lookahead buyers play ``undominated strategies". This is a significantly more robust notion of rationality than a Nash equilibrium. 

\sacommentMinor{removed:
\paragraph{Traditional one round auction definition.} A strategy
$s_i'$ is a dominated strategy for buyer $i$ if $\exists s_i$ such that 
$$\forall v_i, {\bf s}_{-i}:  \ \ \mathbb{E}_{{\bf v}_{-i}}[u_i(v_i, s_i(v_i), {\bf s}_{-i}({\bf v}_{-i}))] \ge \mathbb{E}_{{\bf v}_{-i}} [u_i(v_i, s_i'(v_i), {\bf s}_{-i}({\bf v}_{-i}))]$$
and
$$\exists v_i, {\bf s}_{-i}:  \ \ \mathbb{E}_{{\bf v}_{-i}}[u_i(v_i, s_i(v_i), {\bf s}_{-i}({\bf v}_{-i}))] > \mathbb{E}_{{\bf v}_{-i}} [u_i(v_i, s_i'(v_i), {\bf s}_{-i}({\bf v}_{-i}))]$$
where $u_i(v_i, b_i, {\bf b}_{-i})$ denotes the utility of buyer $i$.
}
\begin{definition}[Dominated strategies for $k$-lookahead buyers.]
In our setting for $k$-lookahead buyer, a strategy $s_i'$ is a dominated strategy at time $t$ under history $H_{t-1}$ if $\exists s_i$ such that 
$$\forall v_{i,t}, {\bf s}_{-i}:  U^{[t,t+k]}_i\bigg( H_{t-1}, {v}_{i,t}, s_i, {\bf s}_{-i}\bigg) \ge U^{[t,t+k]}_i\bigg( H_{t-1}, v_{i,t}, s'_i, {\bf s}_{-i}\bigg) $$
$$\exists v_{i,t}, {\bf s}_{-i}:  U^{[t,t+k]}_i\bigg( H_{t-1}, {v}_{i,t}, s_i, {\bf s}_{-i}\bigg) > U^{[t,t+k]}_i\bigg( H_{t-1}, v_{i,t}, s'_i, {\bf s}_{-i}\bigg) $$
\end{definition}
For myopic buyers, above definition applies with $k=0$.




\subsection{Heterogeneous learning behavior}
We consider learning buyers as those who do not know (or do not trust) the seller's mechanism, in particular the seller's allocation and payment function, in order to be able to precisely evaluate their current and future utility. 
Instead, a learning buyer uses the past outcomes to {\it learn} how to bid. We formalize the notion of learning buyer using the experts learning framework \cite{Freund1995}. 
\sacommentMinor{removed: In general, the bidding strategy used at time $t$ can be an arbitrary mapping from historical information and current valuation $v_{i,t}$ to bid. However, being able to learn the best strategy in such an arbitrary set makes too strong an assumption on learning abilities of the buyer. We consider, as a set of expert strategies, a finite collection $E$ of mappings from the current {\it state} information 
and current valuation to a bid, i.e., set of experts 
\begin{equation}
\label{eq:experts}
E = \{f:[{\cal S}] \times \bar{V} \rightarrow \bar{V}\};
\end{equation}
Here, state is defined as a piece of information determined by the history of the buyer, and may depend on the mechanism. We will use a set of experts defined by 2-bit state for our results. $\bar{V}$ is a  discretized (to arbitrary accuracy) range of valuations, in order to obtain a finite set of experts, and ${\cal S}$ is the set of possible states. Given buyer i's state $S_{i,t}$ and buyer's private valuation $v_{i,t}$, an expert strategy $f\in E$ suggests  bid $b_{i,t}=f(S_{i,t}, v_{i,t})$. 
}
 A learning buyer uses a learning algorithm in order to learn to bid in a way that its total utility is close to that achieved by the best single expert among a set of expert bidding strategies $E$. (Recall from Definition~\ref{def:strategy}, a bidding strategy is an arbitrary mapping from history and valuation to bid). We formalize different levels of learning sophistication among buyers by considering two classes of learning algorithms, as described below. \\ 

\paragraph{No-regret learner:}  A no-regret learning buyer $i$ uses a no-regret learning algorithm to decide bid $b_{i,t}$ at time $t$. 
The `reward' (in no-regret learning terminology) at time $t$ on making a bid $b_{i,t}=b$ is given by the buyer's $t$\textsuperscript{th}-round utility, determined by the  mechanism's output depending on other buyers' bids as well as the history. That is, on making a bid $b$, the learner's reward at time $t$ is given by a function $g_{t}(b)$ defined as 
\[g_{t}(b) :=v_{i,t} x_{i}(H_{t-1}, b, {\bf b}_{-i,t}) - p_{i}({H}_{t-1}, b, {\bf b}_{-i,t}) \] 
Regret is defined as the difference between buyer's total reward and that of the best expert $f \in E$ \emph{in hindsight}:
\begin{equation}
\label{eq:no-regret} 
\text{Regret}(T) = \max_{f\in E}\sum_{t=1}^T g_{t}(f(H_{t-1}, v_{i,t})) - \sum_{t=1}^T g_{t}(b_{i,t}) 
\end{equation}
A no-regret learning buyer uses a bidding strategy such the above regret is $o(T)$ under every trajectory of bids and private valuations.  Note that such a learner is solving an adversarial bandit problem, since the learner only observes the value of function $u_{i,t}(\cdot)$ on the bid $b_{i,t}$ used by the buyer. When the number of experts 
$N$ is finite, there are efficient and natural algorithms (e.g., EXP3 algorithm 
based on multiplicative weight updates \cite{Auer2003}) that achieve $O(\sqrt{NT\log N})$ 
regret. \\

\paragraph{No-policy-regret learner:} This more sophisticated buyer uses a no-{\it policy-regret} learning algorithm (following definition of policy regret in \cite{ADT12}). 
An important distinction from the definition of regret in the previous paragraph is that now the total reward of the best expert must be evaluated over the {\it trajectory} of adversarial inputs (i.e., history and other buyers' bids) in response to the bids made by the expert. To make explicit the dependence of $t^{th}$ utility reward on the trajectory of past decisions through history of outcomes and other buyers' strategic response, let us denote the reward function for round $t$ as $g_{t}(b, H_{t-1})$. Let $s_{j,t}$ denote the strategy used by buyer $j\ne i$ at time $t$ and $v_{j,t}$ is the private valuation of buyer $j$.  Then, the learner's expected reward in round $t$ is defined as:
\[g_{t}(b, H_{t-1}) :=\Ex_{v_i\sim F}[v_i x_{i}(H_{t-1}, b, {\bf b}_{-i,t}) - p_{i}({H}_{t-1}, b, {\bf b}_{-i,t}) \text{ where } b_{j,t} := {\bf s}_{j,t}(H_{t-1}, v_{j,t})].\]
Then, for any sequence of other buyers' valuations ${\bf v}_{-i,t}$ and strategies ${\bf s}_{-i,t}$ for $t=1,\ldots, T$ 
policy-regret of such a buyer is defined against the  best expert  $f\in E$ in hindsight:
\begin{equation}
\label{eq:buyer-policy-regret}
\text{Policy-Regret}(T) = \max_{f\in E} \sum_{t=1}^T g_{t}(f(H'_{t-1},v_{i,t}), {H}'_{t-1}) - \sum_{t=1}^T g_{t}(b_{i,t}, {H}_{t-1})
\end{equation}
where $b_{i,t}$ denotes the bid made by buyer at time $t$,  $f(H'_{i,t},v_{i,t})$ denotes the bid that {\it would} be made by the expert under counterfactual trajectory,  
${H}'_{1}, \ldots, {H}'_{T}$ denotes the (possibly randomized) counterfactual trajectory  of 
history that would be observed in response to using the  bids suggested by the expert, instead of the original bids $b_{i,t}$. A no-policy-regret learning buyer uses a bidding strategy such that the above policy-regret is guaranteed to be $o(T)$ under any sequence of other buyers' valuations ${\bf v}_{i,t}$ and strategies ${\bf s}_{-i,t}$ for $t=1,\ldots, T$. See Appendix~\ref{sec:nprexistence} for a short note on the existence of policy regret learning algorithms.




\section{Repeated First Price Auction Mechanism}
\label{sec:mechanism}
Algorithms~\ref{alg:one}-\ref{alg:good}  contain the formal description of our mechanism.  We design a dynamic first price auction mechanism that is conducted in sequential rounds $t=1, \ldots, T$. In every round, the mechanism partitions the $n$ buyers into two categories: good state  buyers and bad state buyers.  All buyers start in good state. The buyers may be moved by the mechanism from good state to bad state over time but once in bad state, a buyer remains there for the remaining rounds.\footnote{One can think of the time period $T$ as the number of auctions in a day, and reset all buyers to good state at the beginning of the next day when another $T$ auctions are run. Alternatively it also possible to design a mechanism that permits the movement back to a good state for initial few rounds, but for clarity in exposition, so we  choose to not do that here.}  In any round, the current set of buyers in  good state and bad state are denoted by $G$ and $B$ respectively. 
The mechanism uses these states to track buyer behavior and incentivize lookahead or learning buyers to stay in a good state, in order to extract the desired revenue from both sophisticated and naive buyers.  

The mechanism proceeds in epochs, each consisting of multiple rounds. In a given round, the mechanism either conducts a first price auction with a reserve price among the good state buyers, or a first price auction with a (different) reserve price among the bad state buyers.  Specifically, an epoch consists of \mbox{$\epl :=  \eplVal = O(m_g)$} rounds where $\threshold=\Hdelta$ and $m_g = \max(1, |G|)$ is the number of good state buyers in the beginning of the epoch. 
During an epoch, the mechanism first  runs the \badauction \ subroutine (Algorithm~\ref{alg:bad}), which consists of $\rho \epl$ rounds of auctions among bad state buyers. It then runs the \goodauction \ subroutine (Algorithm~\ref{alg:good}),  which consists of $(1- \rho)\epl$ rounds of auctions among good state buyers.  Some good state buyers may be moved to bad state during the \goodauction. The reserve prices $\resg$ and $\resb$ for good and bad state auctions depend on parameters $m_g, m_b \in [n]$ which are set at the beginning of an epoch, and remain fixed throughout an epoch. 


 \begin{algorithm}[H]
	\caption{Robust repeated first price auction among $n$ buyers}
\begin{algorithmic}
	\INPUT horizon $T$,  
	parameters $\rho, \epsilon, \delta \in (0,1)$, thresholds $\threshold=\Hdelta$
	\STATE Initialize all buyers' state to the good state:  $G = \{1, \ldots, n\} \text{ and } B = \{\}$
	\STATE Repeat until horizon of $T$ rounds is reached:
	\STATE \textbf{for} epochs $\ep=1,2, \ldots, 
	$ \textbf{do} \\
	\STATE  \qquad Set $m_b = \max(n/2, |B|)$,  $m_g = \max(1, |G|)$, \mbox{$\epl := \eplVal$}
	\STATE  \qquad  Run  $\badauction$  among buyers in $B$ for $\epl^b=\rho \epl$ rounds
		\STATE \qquad Run  $\goodauction$  among buyers in $G$ for $\epl^g=(1-\rho) \epl$ rounds
		\STATE  \qquad \qquad (this may update $B$ and $G$)
	\end{algorithmic}
	\label{alg:one}
\end{algorithm}

In the first 
$\epl^b=\rho \epl$ rounds of the epoch, \badauction \ runs a first price auction with reserve price $\resb$ among bad state buyers in $B$, as described in Algorithm~\ref{alg:bad}.  
 Here $\resb:=\myeprice-\frac{\epsilon}{n} \qum$, with $\myeprice=F^{-1}(1-\prob_{m_b})$, $\prob_{m}$ being the probability that a buyer wins in a Myerson auction with $m$ buyers (e.g., see \cite{Mye81}); and $\qum:=F^{-1}(1-1/m_b)$ the $m_b^{th}$ quantile. In fact, for any $m$, $\prob_m\le \frac{1}{m}$, so that $\myeprice\ge \qum$ and $\resb \ge (1-\frac{\epsilon}{n}) \myeprice \ge (1-\frac{\epsilon}{n}) \qum$. 

  \begin{algorithm}[H]
	\caption{\badauction: first price auction among bad state buyers}
\begin{algorithmic}
	\INPUT buyers $B$, number of rounds $\epl^b$, and parameter $m_b \in [n]$
		\STATE  \textbf{for}  $\epl^b$ rounds \textbf{do} 
			\STATE  \qquad  \textbf{Auction:} Run a first price auction with reserve price $\resb = \myeprice-\frac{\epsilon}{n} \qum$ among buyers $B$. 
			\STATE \qquad  Let $\{\bidj\}_{j\in B}$ be the set of bids received.
			\STATE  \qquad 	\textbf{if} $\max_{j \in B} \bidj \geq \resb$  \textbf{then}
			   \STATE  \qquad 	     \qquad Allocate the good to buyer $i = \argmax_{j \in B} \bidj$ with highest bid  and charge $\bidi$
	\end{algorithmic}
	\label{alg:bad}
\end{algorithm}

In each of the remaining $\epl^b=(1-\rho) \epl$ rounds of the epoch, the \goodauction \ subroutine (Algorithm~\ref{alg:good}) 
runs a first price auction with reserve price $\resg:= (1-\epsilon)\quplus_{m_g}$ among buyers currently in good state (i.e., buyers in $G$). Here, $\quplus_m:=\Ex_{v\sim F}[v | v\ge q_m]$, with $q_m=F^{-1}(1-1/m)$.

The  \goodauction \ subroutine uses a third state, called the rest state, which is used to temporarily rest a buyer, i.e., not allow that buyer to participate in the  remaining auctions in that epoch. The set of buyers in the rest state in the current round is denoted by $R$. In each round, after the auction a buyer may be moved from good state to either bad or rest state in the following ways:
 \begin{itemize}
     \item The mechanism considers the number of uncleared auctions $U$ so far in this epoch, i.e., the number of auctions during this instance of \goodauction, where all the participating bids were lower than reserve price. If this number is greater than or equal to $\frac{m_g\threshold}{(1-\delta)}$, then any good state buyer ($i\in G$) whose bid $b_{i}$ in this round was smaller than the reserve price, is moved to bad state. 
     \item For every buyer currently in good state ($i\in G$), the mechanism considers the number of allocations $A_i$ received by the buyer so far in this epoch. If this is $\ge \threshold$, the buyer is moved to the rest state. 
 \end{itemize}
\sacommentMinor{added:} In the first step, the mechanism checks if the number of uncleared auctions is significantly above the statistically expected number. And, if so, from there on, the mechanism punishes every participating buyer who bids below reserve price. This step is aimed to ensure that the lookahead buyers are incentivized to bid above reserve price in rounds where their private valuations are high enough. To understand the intuition behind the second step, observe that given the epoch length of $O(m_g)$, statistically, any given buyer is expected to have the highest valuation among $m_g$ buyers for roughly a constant number of steps in every epoch. Thus, the second step is intended to ensure that a buyer does not win too many auctions (perhaps at cost of negative immediate utility for some rounds) in order to deprive other buyers of allocations and potentially cause the mechanism to move them to bad state. 


 \begin{algorithm}[H]
\caption{\goodauction:  first price auction among good state buyers}
\begin{algorithmic}
	\INPUT  buyers $G$ , number of rounds $\epl^g$, parameters $\epsilon, \rho, \delta \in (0,1)$ and $m_g \in [n]$, threshold $\threshold$
\STATE Initialize state $R = \{\}$ and counters    $\uncleared =0$
	 and, for all $i \in G$,   $\allocated_{i}=0$ 
		\STATE  \textbf{for}  $\epl^g$ rounds \textbf{do} 
			\STATE  \qquad  \textbf{Auction:} Run a first price auction with reserve price $\resg = (1-\epsilon)\quplus_{m_g}$ among buyers in $G$
		\STATE  \qquad 	\qquad Let $i = \argmax_{j \in G} \bidj$ be the buyer with highest bid 
			\STATE  \qquad 	\qquad \textbf{if} $b_i\geq \resg$  \textbf{then}
			   \STATE  \qquad 	\qquad         \qquad Allocate the good to buyer $i $ and charge $\bidi$\\
			 \STATE  \qquad 	\qquad           \qquad Update  the number of allocations  to buyer $i$: $\allocated_{i}=\allocated_{i}+1$
	\STATE  \qquad 	\qquad \textbf{else} 
	\STATE  \qquad 	\qquad  \qquad Update the number of uncleared auctions: $\uncleared =\uncleared + 1$
		 \STATE  \qquad 	\textbf{Move buyers between states:}
			\STATE  \qquad \qquad  \textbf{if} $\uncleared \ge \frac{m_g \threshold}{1 - \delta}$  \textbf{then} move buyers  in $G$ with bid lower than $\resg $ to $B$: 
			  \STATE  \qquad \qquad  \qquad  Update $B = B \cup \{i \in G : b_i < \resg\}$ and  $G=G\backslash\{i : b_i < \resg\}$
			\STATE  \qquad \qquad \textbf{if}  $\allocated_{i} \ge \threshold$  \textbf{then} move buyer $i$ to rest state:
		       \STATE  \qquad \qquad  \qquad   Update $R=R\cup\{i\}$ and  $G=G\backslash\{i\}$

		Move all rest state buyers back to good state: $G=G\cup R$
	\end{algorithmic}
	\label{alg:good}
\end{algorithm}

\section{Revenue Analysis: Main Result}

Our main result is that the mechanism presented in the previous section extracts a constant factor of optimal revenue  $\quplus_{\nsoph} + \revmye(\nnaive)$ from $\nsoph$ sophisticated buyers and $\nnaive$ naive buyers.  To formally state this result, we first  define the sophisticated buyers and naive buyers. This involves defining the set of experts used by learning buyers. In general, an expert bidding strategy can be any arbitrary mapping from historical information and current valuation  to bid. However, being able to learn the best strategy in such an arbitrary set makes too strong an assumption on learning abilities of the buyer. In fact, it is sufficient for our mechanism to have learning buyers that can compete against a restricted set of experts, as defined below. 

\begin{definition}[Expert set $E$]
\label{def:E}
Let $h_{i,t-1}$ be a fixed size projection of history $H_{t-1}$ containing the following information: the buyer $i$'s state (whether it is bad or good/rest) in round $t$, the number of buyers in good and bad state, and the number of uncleared auctions so far in the current epoch. Let ${\cal H}$ denote the set of ($2\times n \times \eplmax$) possible values of $h_{i,t}, \forall i,t$. Then, set of experts $E$ is defined as mappings from this projected history 
and current valuation to a bid, i.e., $E = \{f:{\cal H} \times \bar{V} \rightarrow \bar{V}\}.$

Here, $\bar{V}$ is a  discretized (to arbitrary accuracy) range of valuations, in order to obtain a finite set of experts. 
Given $h_{i,t-1}$, and valuation $v_{i,t}$, an expert strategy $f\in E$ suggests  bid $b_{i,t}=f(h_{i,t-1}, v_{i,t})$ to buyer $i$ in round $t$. 
\end{definition}

We are now ready to formally define sophisticated and naive buyers.

\begin{definition}[Sophisticated and Naive buyers]
Sophisticated buyers are defined as the buyers who are   either $k$-lookahead, for $k \ge \frac{80\log(\epsilon^{-1})n}{\epsilon^3(1-\epsilon)^2(1 - \rho)} = \Theta(n)$, or no-policy-regret learners against some set of experts containing $E$.  Na\"{i}ve buyers are defined as buyers who are either myopic, or are no-regret learners against some set of experts containing $E$.
\end{definition}

\begin{restatable}{rThm}{thmmain}
\label{thm:main}
Assuming all myopic and $k$-lookahead buyers play undominated strategies, 
 the expected per round revenue of the mechanism described in Algorithm~\ref{alg:one}-\ref{alg:good}, with $\epsilon\in (0,1)$, $\delta = \epsilon$,  $\rho \le\rhoCompactBound$, is at least
$$\Theta(1) \left( \quplus_{\nsoph}  +  \Rev^{\Mye}(\nnaive) \right) -o(1),$$
where $\nsoph, \nnaive$ is the number of sophisticated and naive buyers, respectively. 
More precisely, the expected per round revenue is at least
   $$(1-\epsilon)\frac{1}{4}\cdot  \quplus_{\nsoph}  +\frac{\rho(1-\epsilon)}{2}  \left(1-\frac{1}{e}\right)\Rev^{\Mye}(\nnaive)  -o(1)$$
 where $\quanplus=\Ex_{v\sim F}[v|v\ge \quan]$, with $\quan=F^{-1}(1-1/n)$ being the $n^{th}$ quantile for the valuation distribution, and $\quplus_0 = 0$.  $\Rev^{\Mye}(n)$ is the optimal revenue in a single-item auction with $n$ buyers.
\end{restatable}

The proof of Theorem~\ref{thm:main} consists of four parts.
 We first give revenue and utility bounds that apply to any buyer, then characterize  undominated strategies for myopic/lookahead buyers, and no-regret strategies for learners, and finally combine these parts. We give an overview  of each  part here and defer lemma statements and their proofs to the appendix.

\paragraph{General revenue and utility analysis (Appendix~\ref{sec:appgeneral}).}   
 We first establish in Lemma~\ref{lem:revGood} that, if $G$ is the set of good state buyers at the end of an epoch, then the expected  revenue from \goodauction \ during that epoch is at least  $|G| \threshold  \resg$. This lower bound on revenue from good state buyers is obtained by observing that each  buyer who ends an epoch in good state must have either been allocated  the item (and paid at least $\resg$) during at least $\threshold$ rounds of this epoch to be moved to a rest state, or must have bid at least $\resg$ at each round where $\uncleared \geq m_g H/(1 -\delta)$. In Lemma~\ref{lem:revenuebad}, we show that if the bid of bad state buyers $B$ is at least  $\resb$ when their value is larger than $\resb$, then the expected revenue per round of \badauction \ is at least $(1 - \epsilon)(1-1/e)\frac{|B|}{m_b} \revmye(m_b)$.

 A main part of the overall revenue analysis is to argue that sophisticated buyers are incentivized to remain in  good state, irrespectively of other buyers' bids. To show this, we establish a lower bound on the utility achievable in good state and an upper bound  on the utility achievable in bad state. To establish the lower bound, we analyze a strategy called the good strategy  $s^g$ (Definition~\ref{def:goodstrategy}) that never moves a buyer to the bad state. When $U < m_g H / (1 - \delta)$, $s^g$ bids $\resg=(1-\epsilon)\quplus_{m_g}$ if $v_{i}\ge \qu_{m_g}$, and  $0$ otherwise. When $U \geq m_g H / (1 - \delta)$, $s^g$ bids $\resg$. In Lemma~\ref{lem:lbutility}, we lower bound the expected utility obtained by strategy $s^g$ over an epoch. The crucial component of the mechanism which allows this bound is temporarily moving buyers who have already been allocated enough (more than statistically expected) number of items to the rest state. This temporarily removes such buyers from good state auctions, and guarantees to any buyer $i \in G$  a minimum number of rounds in each epoch where $i$ can get the item if it bids above the reserve price. In Lemma~\ref{lem:ubutility}, 
 we upper bound the expected utility achievable by any bad state buyer.

\paragraph{Undominated strategies for buyers with heterogeneous lookahead attitudes (Appendix~\ref{sec:lookahead}).} The main lemma  for this part (Lemma~\ref{lem:dominating}) shows that a $k$-lookahead buyer, for $k$ large enough, never enters the bad state. To show this, we consider a round where a buyer $i \in G$ faces the threat to be sent to a bad state if it bids below $\resg$. In Lemma~\ref{lem:lbutilitylookahead}, we lower bound the $k$-lookahead utility obtained by strategy $s^g$, which maintains $i$ in good state, in such a round.  Lemma~\ref{lem:ubutilitylookahead} then upper bounds the $k$-lookahead utility of any strategy bidding below $\resg$ in such a round, which would send the buyer to a bad state.
Lemma~\ref{lem:dominating} then combines these two bounds to show that any strategy that sends a $k$-lookahead buyer to a bad state is dominated by strategy $s^g$. A main difficulty in combining these two lemmas is that the epoch lengths and the reserve prices vary at each epoch, and we need to compare utilities from different epochs.
We show in Lemma~\ref{lem:myopic} that a myopic buyer bids at least the reserve price $\resb$ when it has value at least $\resb$ in bad state.

\paragraph{Strategies of no-regret buyers with heterogeneous learning behaviors (Appendix~\ref{sec:learning}).} We show in Lemma~\ref{lem:policy} that a buyer that goes to a bad state has high policy-regret compared to an expert that plays strategy $s^g$, which implies that a no-policy regret learner must remain in good state in all but $o(T)$ rounds. A difficulty here is to argue that there is gap between the utility a buyer going to a bad state and the utility of an expert following the good strategy, where the utilities are evaluated over different trajectories of adversarial inputs.
In Lemma~\ref{lem:noregret}, we give a condition under which a no-regret learner in bad state must bid at least the reserve price $\resb$ when its value is larger than $\resb$ in all but $o(T)$ rounds. An important subtlety for no-regret learners is that due to the other buyers, a learner is not guaranteed to win a bad state auction and obtain positive utility when it bids at least $\resb$.

\paragraph{Main result (Appendix~\ref{sec:mainresult}).} We combine the three previous parts  to lower bound the revenue achieved by the mechanism and obtain Theorem~\ref{thm:main}. A last non-trivial argument needed is that if a naive buyer remains in good state, we  obtain at least a much revenue from that buyer as if it was in bad state, regardless of how many buyers are in good and bad state.

\newpage 

\bibliographystyle{alpha}
\bibliography{klookahead}

\begin{thebibliography}{ACBFS03}

\bibitem[ACBFS03]{Auer2003}
Peter Auer, Nicol\`{o} Cesa-Bianchi, Yoav Freund, and Robert~E. Schapire.
\newblock The nonstochastic multiarmed bandit problem.
\newblock {\em SIAM J. Comput.}, 32(1):48--77, January 2003.

\bibitem[ada]{adageGoogle}
adage.com.
\newblock Google's ad manager will move to first-price auction.
\newblock
  \url{https://adage.com/article/digital/google-adx-moving-a-price-auction/316894/}.

\bibitem[ade]{rubicon}
adexchanger.com.
\newblock Rubicon joins first-price auction club; diageo is latest brand to
  demand more transparency.
\newblock \url{https://adexchanger.com/ad-exchange-news/tuesday-12122017//}.

\bibitem[ADH16]{ADH16}
Itai Ashlagi, Constantinos Daskalakis, and Nima Haghpanah.
\newblock Sequential mechanisms with ex-post participation guarantees.
\newblock In {\em Proceedings of the 2016 {ACM} Conference on Economics and
  Computation, {EC} '16, Maastricht, The Netherlands, July 24-28, 2016}, pages
  213--214, 2016.

\bibitem[ADMS18]{ADMS18}
Shipra Agrawal, Constantinos Daskalakis, Vahab~S. Mirrokni, and Balasubramanian
  Sivan.
\newblock Robust repeated auctions under heterogeneous buyer behavior.
\newblock In {\em Proceedings of the 2018 {ACM} Conference on Economics and
  Computation, Ithaca, NY, USA, June 18-22, 2018}, page 171, 2018.

\bibitem[ADT12]{ADT12}
Raman Arora, Ofer Dekel, and Ambuj Tewari.
\newblock Online bandit learning against an adaptive adversary: from regret to
  policy regret.
\newblock In {\em ICML}. icml.cc / Omnipress, 2012.

\bibitem[AS13]{AS13}
Susan Athey and Ilya Segal.
\newblock An efficient dynamic mechanism.
\newblock {\em Econometrica}, 81(6):2463--2485, 2013.

\bibitem[Bat05]{Battaglini05}
Marco Battaglini.
\newblock Long-term contracting with markovian consumers.
\newblock {\em American Economic Review}, 95(3):637--658, 2005.

\bibitem[BB84]{BB84}
David~P. Baron and David Besanko.
\newblock Regulation and information in a continuing relationship.
\newblock {\em Information Economics and Policy}, 1(3):267 -- 302, 1984.

\bibitem[Bes85]{Besanko85}
David Besanko.
\newblock Multi-period contracts between principal and agent with adverse
  selection.
\newblock {\em Economics Letters}, 17(1–2):33 -- 37, 1985.

\bibitem[BLP06]{BLP06}
Moshe Babaioff, Ron Lavi, and Elan Pavlov.
\newblock Single-value combinatorial auctions and implementation in undominated
  strategies.
\newblock In {\em Proceedings of the Seventeenth Annual ACM-SIAM Symposium on
  Discrete Algorithm}, SODA '06, pages 1054--1063, 2006.

\bibitem[BML17]{BMP16}
Santiago~R. Balseiro, Vahab~S. Mirrokni, and Renato~Paes Leme.
\newblock Dynamic mechanisms with martingale utilities.
\newblock In {\em Proceedings of the 2017 {ACM} Conference on Economics and
  Computation, {EC} '17, Cambridge, MA, USA, June 26-30, 2017}, page 165, 2017.

\bibitem[BMLZ19]{BMLZ19}
Santiago~R. Balseiro, Vahab~S. Mirrokni, Renato~Paes Leme, and Song Zuo.
\newblock Dynamic double auctions: Towards first best.
\newblock In {\em Proceedings of the Thirtieth Annual {ACM-SIAM} Symposium on
  Discrete Algorithms, {SODA} 2019, San Diego, California, USA, January 6-9,
  2019}, pages 157--172, 2019.

\bibitem[BMSW18]{BMSW18}
Mark Braverman, Jieming Mao, Jon Schneider, and Matthew Weinberg.
\newblock Selling to a no-regret buyer.
\newblock In {\em Proceedings of the 2018 {ACM} Conference on Economics and
  Computation, Ithaca, NY, USA, June 18-22, 2018}, pages 523--538, 2018.

\bibitem[BS15]{BS15}
Dirk Bergemann and Philipp Strack.
\newblock Dynamic revenue maximization: A continuous time approach.
\newblock {\em Journal of Economic Theory}, 159, Part B:819 -- 853, 2015.
\newblock Symposium Issue on Dynamic Contracts and Mechanism Design.

\bibitem[CDKS16]{CDKS16}
Shuchi Chawla, Nikhil~R. Devanur, Anna~R. Karlin, and Balasubramanian Sivan.
\newblock Simple pricing schemes for consumers with evolving values.
\newblock In {\em Proceedings of the Twenty-Seventh Annual {ACM-SIAM} Symposium
  on Discrete Algorithms, {SODA} 2016, Arlington, VA, USA, January 10-12,
  2016}, pages 1476--1490, 2016.

\bibitem[CH00]{CL00}
Pascal Courty and Li~Hao.
\newblock Sequential screening.
\newblock {\em Review of Economic Studies}, 67(4):697--717, 2000.

\bibitem[CHMS10]{CHMS10}
Shuchi Chawla, Jason~D. Hartline, David~L. Malec, and Balasubramanian Sivan.
\newblock Multi-parameter mechanism design and sequential posted pricing.
\newblock In {\em Proceedings of the 42nd {ACM} Symposium on Theory of
  Computing, {STOC} 2010, Cambridge, Massachusetts, USA, 5-8 June 2010}, pages
  311--320, 2010.

\bibitem[dig]{fpera}
digiday.com.
\newblock Programmatic advertising is preparing for the first-price auction
  era.
\newblock
  \url{https://digiday.com/marketing/programmatic-advertising-readying-first-price-auction-era//}.

\bibitem[dKaS92]{KRS92}
J.~de~Kleer and O.~Raiman andMark Shirley.
\newblock One step lookahead is pretty good.
\newblock {\em Readings in Model-Based Diagnosis}, pages 138--142,, 1992.

\bibitem[DPS15]{DPS15}
Nikhil~R. Devanur, Yuval Peres, and Balasubramanian Sivan.
\newblock Perfect bayesian equilibria in repeated sales.
\newblock In {\em Proceedings of the Twenty-Sixth Annual {ACM-SIAM} Symposium
  on Discrete Algorithms, {SODA} 2015, San Diego, CA, USA, January 4-6, 2015},
  pages 983--1002, 2015.

\bibitem[ES07]{ES07}
Peter Eso and Balázs Szentes.
\newblock Optimal information disclosure in auctions and the handicap auction.
\newblock {\em Review of Economic Studies}, 74(3):705--731, 2007.

\bibitem[FS95]{Freund1995}
Yoav Freund and Robert~E. Schapire.
\newblock A decision-theoretic generalization of on-line learning and an
  application to boosting.
\newblock In {\em Proceedings of the Second European Conference on
  Computational Learning Theory}, EuroCOLT '95, pages 23--37, London, UK, UK,
  1995. Springer-Verlag.

\bibitem[FVB06]{FV06}
Drew Fudenberg and J~Miguel Villas-Boas.
\newblock Behavior-based price discrimination and customer recognition.
\newblock {\em Handbook on economics and information systems}, 1:377--436,
  2006.

\bibitem[ILPT17]{ILPT17}
Nicole Immorlica, Brendan Lucier, Emmanouil Pountourakis, and Samuel Taggart.
\newblock Repeated sales with multiple strategic buyers.
\newblock In {\em Proceedings of the 2017 {ACM} Conference on Economics and
  Computation, {EC} '17, Cambridge, MA, USA, June 26-30, 2017}, pages 167--168,
  2017.

\bibitem[KLN13]{KLN13}
Sham~M Kakade, Ilan Lobel, and Hamid Nazerzadeh.
\newblock Optimal dynamic mechanism design and the virtual-pivot mechanism.
\newblock {\em Operations Research}, 61(4):837--854, 2013.

\bibitem[LP17]{LiuP17}
Siqi Liu and Christos{-}Alexandros Psomas.
\newblock On the competition complexity of dynamic mechanism design.
\newblock {\em CoRR}, abs/1709.07955, 2017.

\bibitem[MLTZ16]{MPTZ16a}
Vahab~S. Mirrokni, Renato~Paes Leme, Pingzhong Tang, and Song Zuo.
\newblock Dynamic auctions with bank accounts.
\newblock In {\em Proceedings of the Twenty-Fifth International Joint
  Conference on Artificial Intelligence, {IJCAI} 2016, New York, NY, USA, 9-15
  July 2016}, pages 387--393, 2016.

\bibitem[MLTZ18]{MLTZ18}
Vahab~S. Mirrokni, Renato~Paes Leme, Pingzhong Tang, and Song Zuo.
\newblock Non-clairvoyant dynamic mechanism design.
\newblock In {\em Proceedings of the 2018 {ACM} Conference on Economics and
  Computation, Ithaca, NY, USA, June 18-22, 2018}, page 169, 2018.

\bibitem[MTV12]{MTV12}
Vahab~S. Mirrokni, Nithum Thain, and Adrian Vetta.
\newblock A theoretical examination of practical game playing: Lookahead
  search.
\newblock In {\em Algorithmic Game Theory - 5th International Symposium, {SAGT}
  2012}, pages 251--262, 2012.

\bibitem[Mye81a]{M81}
R.~Myerson.
\newblock Optimal auction design.
\newblock {\em Mathematics of Operations Research}, 6:58--73, 1981.

\bibitem[Mye81b]{Mye81}
Roger~B. Myerson.
\newblock Optimal auction design.
\newblock {\em Math. Oper. Res.}, 6(1):58--73, February 1981.

\bibitem[Nau83]{N83}
Dana~S. Nau.
\newblock Decision quality as a function of search depth on game trees.
\newblock {\em J. {ACM}}, 30(4):687--708, 1983.

\bibitem[PPPR16]{PPPR16}
Christos Papadimitriou, George Pierrakos, Christos-Alexandros Psomas, and Aviad
  Rubinstein.
\newblock On the complexity of dynamic mechanism design.
\newblock In {\em Proceedings of the Twenty-seventh Annual ACM-SIAM Symposium
  on Discrete Algorithms}, SODA '16, pages 1458--1475, Philadelphia, PA, USA,
  2016. Society for Industrial and Applied Mathematics.

\bibitem[PST14]{PST14}
Alessandro Pavan, Ilya Segal, and Juuso Toikka.
\newblock Dynamic mechanism design: A myersonian approach.
\newblock {\em Econometrica}, 82(2):601--653, 2014.

\bibitem[Sim55]{Sim55}
Herbert~A. Simon.
\newblock A behavioral model of rational choice.
\newblock 69(1):99--118, 1955.

\bibitem[SKN09]{SKN09}
E.~Sefer, U.~Kuter, and D.~Nau.
\newblock Real-time a* search with depth-k lookahead.
\newblock In {\em Proceedings of the International Symposium on Combinatorial
  Search}, 2009.

\bibitem[SW94]{SW94}
Dale Stahl and Paul Wilson.
\newblock Experimental evidence on players' models of other players.
\newblock {\em Journal of Economic Behavior \& Organization}, 25(3):309--327,
  1994.

\end{thebibliography}

\newpage
\appendix


\section{Proof for Revenue Upper Bound}
\label{sec:rev_up}

\paragraph{Upper bound on revenue.} In a setting with just $\nnaive$ $\naive$ buyers and no other buyers (for e.g., just $\nnaive$ myopic buyers), it is impossible to get more than $\revmye(\nnaive)$ per round, i.e., the optimal revenue in a single-round auction obtainable from Myerson's auction~\cite{M81}.
In a setting with just $\nsoph$ sophisticated buyers (and no other buyers), it is impossible to get more than $\E_{v_1\sim F,\dots,v_{\nsoph}\sim F}\left[\max(v_1,\dots,v_{\nsoph})\right]$ as revenue is upper bounded by the maximum valuation. It is easy to show that $\E_{v_1\sim F,\dots,v_{\nsoph}\sim F}\left[\max(v_1,\dots,v_{\nsoph})\right] \leq \quansplus =  \E_{v\sim F}\left[v | v \geq F^{-1}(1-1/\nsoph)\right]$. To see this, note that because a buyer has the largest value among $\nsoph$ buyers (with ties broken uniformly) with probability $\frac{1}{\nsoph}$, $$\E_{v_1\sim F,\dots,v_{\nsoph}\sim F}\left[\max(v_1,\dots,v_{\nsoph})\right] = \sum_{i=1}^{\nsoph} \frac{1}{\nsoph} \E_{v_1\sim F,\dots,v_{\nsoph}\sim F} \left[v_i| v_i=\max(v_1,\dots,v_{\nsoph})\right].$$ Now, since the expected value of a buyer conditioned on an event happening with probability $1/\nsoph$ is at most $\E_{v\sim F}\left[v | v \geq F^{-1}(1-1/\nsoph)\right] = \quansplus$, the inequality follows.

Now, combining these, we claim that in  a setting with $\nsoph$ sophisticated buyers \emph{and} $\nnaive$ $\naive$ buyers, the total revenue achievable is at most $\quansplus + \revmye(\nnaive)$. Indeed, if we were able to achieve more than this, then either the revenue contribution from the sophisticated buyers is more than $\quansplus$ or the $\naive$ buyers is more than $\revmye(\nnaive)$ --- neither of this is possible because if that was true then in a setting with just the $\nsoph$ sophisticated buyers or just the $\nnaive$ $\naive$ buyers we could have simulated the rest of the buyers by adding dummy buyers and discarded the revenue contributed by dummy buyers to obtain more revenue than $\quansplus$ or $\revmye(\nnaive)$.

\section{General Revenue and Utility Analysis}
\label{sec:appgeneral}


The analysis shows that, up to constant factors, the mechanism extracts the optimal $\quplus_{\nsoph} + \revmye(\nnaive)$ revenue from $\nsoph$ sophisticated buyers and $\nnaive$ na\"ive buyers. In Section~\ref{sec:revenue},  we first provide separate bounds on the revenue from good and bad state buyers in Lemma \ref{lem:revGood} and Lemma \ref{lem:revenuebad} respectively. 

The main part of the analysis is to argue that sophisticated buyers, either $k$-lookahead buyers for large enough $k$ or no-policy-regret learners, are incentivized to remain in the good state. We show that any strategy which leads to the bad state is a dominated strategy for $k$-lookahead buyers (and has large regret for no-policy regret learners). To show this, in Section~\ref{sec:utility},  we provide lower and upper bounds on the utility achievable by a buyer in the good and bad state in Lemma~\ref{lem:lbutility} and Lemma~\ref{lem:ubutility},  respectively. 

These utility bounds are used Section~\ref{sec:lookahead} and Section~\ref{sec:learning} to argue that (a) $k$-lookahead buyers for large enough $k$ (Lemma \ref{lem:dominating}) and no-policy-regret buyers (Lemma~\ref{lem:policy}) have incentive to stay in good state for most rounds, and (b) in bad state, myopic buyers (Lemma \ref{lem:myopic}) and learning buyers (Lemma~\ref{lem:noregret}) have incentive to bid above reserve price when their private valuation is large enough. 

Finally in Section~\ref{sec:mainresult}, we combine all these observations to lower bound the revenue achieved given a pool of heterogeneous lookahead and learning buyers, and prove our main result.



\subsection{Revenue analysis from good and bad state auctions}
\label{sec:revenue}

We give bounds on the revenue achieved by \goodauction \ and \badauction. These are general bounds which hold for both lookahead and learning buyers. 

  \begin{lemma}
  \label{lem:revGood} Let $G$ be the set of good state buyers at the end of an epoch. Then, total expected revenue from \goodauction \ during that epoch  is at least  $|G| \threshold (1 - \epsilon) \quplus_{m_g}$.
  \end{lemma}
\begin{proof}  Let $G$  be the good state buyers at the end of an epoch of the mechanism.  Consider the round during that epoch at which the number $U$ of uncleared auctions reaches $\frac{m_g \threshold}{1-\delta}$. Such a step must exist because total number of allocations is at most $m_g \threshold$ ($m_g$ is  the number of good state buyers at the beginning of that epoch) but the number of good state auctions in this epoch is $(1 - \rho) \epl =  \frac{2m_g\threshold}{1-\delta}$. Suppose that at this round, some buyer in $G$ is in good state (i.e., hasn't yet been moved to rest state). Every such buyer must  bid $\resg$ or above in each of the remaining rounds of the epoch, until that buyer is moved to the rest state; otherwise, the mechanism would have pushed this buyer to a bad state and this buyer would not be in $G$ at the end of this epoch. This means that before the end of the epoch: either all of the buyers in $G$ were moved to rest state so that revenue was at least $|G| \threshold \resg$; or all the remaining auctions (after the number of uncleared auctions reached the threshold) cleared because some bid exceeded reserve price, so that the number of uncleared auctions is bounded by the threshold $\frac{m_g\threshold}{1-\delta}$, and the number of goods sold through the good state auctions is at least $(1 - \rho)\epl -\frac{m_g\threshold}{1-\delta} \geq m_g \threshold$,  giving revenue of at least $\resg |G| \threshold$. 
\end{proof}

\paragraph{Revenue from bad state auctions.}

\begin{lemma}	
 \label{lem:revenuebad}
Let $B$ be set of bad state buyers in the beginning of an epoch. Suppose that in every round of \badauction during this epoch where the set of buyers $i\in B$ with $v_{i,t} \geq \myeprice$ is non-empty, at least one such buyer is guaranteed to bid $\resb$ or more. 
Then, total expected revenue from that epoch is at least    $(1-\epsilon)(1-\frac{1}{e}) \cdot \frac{|B|}{m_b}  \Rev^{\Mye}(m_b) \epl^b$, where expectation is taken over valuations $v_{i,t}, i\in B, t\in \epl^b$. 
\end{lemma}
\begin{proof}
Recall $\resb =\myeprice - \frac{\epsilon}{n} \qum \ge (1-\epsilon)\myeprice$. 
If in every round, among buyers in bad state with valuation above $\myeprice$, at least one buyer is guaranteed to bid above $\resb\ge (1-\epsilon)\myeprice$, then the mechanism will get at least $(1-\epsilon)$ fraction of the expected revenue of a posted-price mechanism with a uniform price $\myeprice$ among buyers in $B$. Since  $\theta_{m_b}$ is the probability of a buyer winning an $m_b$ buyer Myerson's auction with iid values, where $m_b\ge |B|$, this revenue (see for example~\cite{CHMS10}) is at least $(1-\epsilon)(1-1/e) \frac{|B|}{m_b}\cdot \Rev^{\Mye}(m_b)$ for every round in bad state auctions. 
\end{proof}

\subsection{Utility analysis from good and bad state buyers}
\label{sec:utility}

 We give bounds on the utility achievable by a buyer in a good state and bad state. These are general bounds which hold both for buyers with lookahead attitudes and the buyers with learning behaviors. Further, they hold irrespective of other buyers' bids.

\paragraph{Lower bound on the utility achievable by a good state buyer.} 

We lower bound the utility achievable by a buyer in the good state by describing a simple strategy which (1) never moves a buyer to the bad state and (2) achieves high utility. We call this strategy the `good strategy', denoted by $s^g$, and defined as follows. 
\begin{definition}[Good strategy $s^g$]
\label{def:goodstrategy}
In any round of \goodauction, a buyer $i$ using good strategy $s^g$ bids in the following manner. If the number of the uncleared past auctions in the current epoch is $U < m_g H / (1 - \delta)$, bid $\resg=(1-\epsilon)\quplus_{m_g}$ if the current valuation $v_{i}\ge \qu_{m_g}$, and bid $0$, otherwise. If $U \geq m_g H / (1 - \delta)$, set the bid equal to the reserve price $\resg$ irrespective of the current valuation. 
\end{definition}
Observe that the good strategy is defined in such a manner that a buyer using this strategy for all rounds of \goodauction is guaranteed to either always stay in the good state or be moved to the rest state, i.e., is guaranteed to be never pushed to a bad state. Next, we now lower bound the utility of a buyer using this strategy.

\begin{restatable}{rLem}{lemlbutility}
\label{lem:lbutility} Consider a buyer in good state at the beginning of an epoch. 
Irrespective of the other buyers' bids, the expected utility of strategy $s^g$ over that epoch  is at least 
$$\epsilon (1-\epsilon) \quplus_{m_g} H.$$
\end{restatable}
\begin{proof} 

Consider an epoch where a good state buyer follows strategy $s^g$ for all rounds in \goodauction.
Similarly to the proof of Lemma~\ref{lem:revGood}, observe that the number of uncleared auction in such an epoch must reach $\frac{1}{(1-\delta)}m_g\threshold$ before the end of this epoch.
Let $e_1$ denote the event that the buyer is still in good state when the number of uncleared auction in that epoch reaches $\frac{1}{(1-\delta)}m_g\threshold$, i.e., the buyer did not get enough allocations to go to rest state. Since the buyer bids $\resg=(1-\epsilon)\quplus_{m_g}$ whenever $v\ge \qu_{m_g}$, probability $\Pr({e_1})$ of this event happening is bounded by the probability of the following event: let $X_1, X_2, \ldots, X_r$ be $r=\frac{m_g\threshold}{1-\delta}$ independent samples from distribution $F$; then  consider event $\sum_{i=1}^r I(X_i \ge F^{-1}(1-1/m_g)) \le \threshold$. Since $\Ex[\sum_{i=1}^r I(X_i \ge F^{-1}(1-1/m_g))] = \frac{r}{m_g}$, Therefore, using Chernoff bounds, probability of this event is bounded as $\Pr({e_1})\le e^{-\delta^2\threshold/2}$. Under event $e_1$, the buyer may end up with negative utility (as the strategy of always bidding reserve price will kick in), which can be at worst $-\resg \threshold$ from this epoch. Otherwise, the buyer will get $H$ allocations, each from some rounds where $v\ge \qu_{m_g}$. Since the buyer wins them at reserve price, and the space of other buyers' bidding strategies consist only of functions that are independent of this buyer's bid and valuation, the expected utility from each of these $H$ goods is $\Ex[v-(1-\epsilon)\quplus_{m_g} | v\ge \qu_{m_g}] \ge \epsilon \quplus_{m_g}$. 

Then, expected utility from each of of the next $K-1$ epochs is at least:
\begin{eqnarray}
& & (1-\Pr(e_1)) \threshold \epsilon \quplus_{m_g} + \Pr(e_1) \threshold (-\resg)\nonumber\\
& \ge & (1-e^{-\delta^2\threshold/2}) \epsilon \quplus_{m_g} \threshold  -  e^{-\delta^2\threshold/2} (1-\epsilon) \quplus_{m_g} \threshold \nonumber\\
(\text{substituting $\delta=\deltaH$}) & \ge &  \epsilon \quplus_{m_g} \threshold  -  \epsilon^2(1-2\epsilon ) \quplus_{m_g} \threshold\nonumber\\
& \ge &  \epsilon (1-\epsilon) \quplus_{m_g} \threshold  
\end{eqnarray}
\end{proof}

\paragraph{Upper bound on utility achievable by a bad state buyer.} 
\begin{lemma}
\label{lem:ubutility}
Consider a buyer in bad state at the beginning of an epoch. 
Irrespective of the other buyers' bids,  its expected utility over the epoch  is at most 
$$  (1+\epsilon) \frac{\rho \epl}{m_b}  \quplus_{m_b}.$$
\end{lemma}
\begin{proof}
For an epoch of length $\epl$, the number of rounds of \badauction \ during an epoch is $\rho \epl$. During a bad state auction, the utility of a buyer with value $v$  is at most $(v-\resb)I(v\ge \resb)$. Thus, the total expected utility of a buyer in bad state over the epoch is at most 
		\begin{eqnarray*}
	\label{eq:1}
&	&\rho \epl \cdot \Ex[(v-\resb)I(v\ge \resb)] \\
& = &   \rho \epl \cdot \Ex[(v-\resb)I(v\ge \myeprice)] + \rho \epl \cdot \Ex[(v-\resb) I(\myeprice \ge v\ge \myeprice - \frac{\epsilon}{n}\qum)]\\
& \le &   \rho \epl \cdot \Ex[(v-\resb)I(v\ge \qu_{m_b})] + \rho \epl \cdot \left(\frac{\epsilon}{n} \qum\right) \\
 & \le &   \rho \epl  \frac{1}{m_b}  (1+\epsilon)\quplus_{m_b} \nonumber
	\end{eqnarray*}
 where the first and second inequality were obtained using that $\resb=\myeprice-\frac{\epsilon}{n} \qum$, and $\qum\le \myeprice$. 
The first part of third inequality used $v\le v-\resb$, and by definition $\quplus_{m_b} =\Ex[v|v\ge \qu_{m_b}]$, $\Pr(v\ge \qu_{m_b}) =\frac{1}{m_b} $. The second part used  $\qum \le \quplus_{m_b}, n \ge m_b$.
\end{proof}


\section{Undominated Strategies for Buyers with Heterogeneous Lookahead Attitudes}
\label{sec:lookahead}

	We characterize the undominated strategies of myopic and lookahead buyers. The main lemma in this section, Lemma~\ref{lem:dominating}, argues that a $k$-lookahead buyer, for $k$ large enough, playing an undominated strategy never enters a bad state. Lemma~\ref{lem:myopic} then shows that in a bad state, myopic  buyers playing undominated strategies always bid at least the reserve price $\resb$ when their value is at least $\resb$. 

\paragraph{Lookahead buyers.} We previously showed in Lemma~\ref{lem:revGood} that we obtain the optimal $|G| \quplus_{m_g}$ revenue from buyers in a good state. In addition to obtaining the optimal revenue, we  show that the mechanism also incentivizes lookahead buyers to stay in a good state. We denote the maximum length of an epoch by $\eplmax = \max_{m_g} \eplVal = \eplmaxVal$.

\begin{lemma}
\label{lem:dominating}
A $k$-lookahead buyer with $k\ge \Kbound  \eplmax$ playing an undominated strategy never  enters the bad state.
\end{lemma}

The mechanism moves a buyer $i$ to the bad state if a buyer bids lower than $r^g$  in a round of \goodauction \ where the number of previous uncleared auctions is $U \geq m_g H / (1 - \delta)$. We show that during a round of  \goodauction \ with $U \geq m_g H / (1 - \delta)$, for every $k$-lookahead buyer $i$ in good state $G$,  bidding below the  reserve price $r^g$ is a dominated strategy.

 A strategy which dominates bidding below the reserve price in such a case is the good strategy $s^g$ from Definition~\ref{def:goodstrategy}. Recall that  if $U < m_g H / (1 - \delta)$, this strategy bids $\resg=(1-\epsilon)\quplus_{m_g}$ if the valuation $v_{i}$ is greater than or equal to $\qu_{m_g}$, and  $0$, otherwise; if the number of uncleared past auctions in this epoch is $\frac{1}{(1-\delta)} m_g \threshold$ or more, it sets the bid equal to the reserve price $\resg$ for every value. Under this strategy, the buyer is guaranteed to always stay in good state or be moved to rest state, i.e., guaranteed to never be pushed to a bad state. 

Lemma~\ref{lem:ubutilitylookahead} and Lemma~\ref{lem:lbutilitylookahead}  upper and lower bound the $k$-lookahead utility achieved by bidding below the reserve price and by playing $s^g$ during a round of  \goodauction \ with $U \geq m_gH / (1 - \delta)$. Then, by combining these bounds, we get that bidding below the reserve price in such a situation is dominated and obtain Lemma~\ref{lem:dominating}.  We denote by $\ep(t)$ the epoch at which step $t$ occurs.

\begin{lemma}
\label{lem:ubutilitylookahead}
Consider a round $t$ during \goodauction \    where the number of previous uncleared auctions is $U \geq nH / (1 - \delta)$. Then, for any buyer $i \in G$ and strategy $s_i$ which bids lower than $r^g$ in this round, 
$$U_i^{[t, t+k]}(H_{t-1}, v_{i,t}, {\bf s}) \leq \frac{2(1+\epsilon)\rho \threshold }{(1-\rho)(1-\delta)} \sum_{j=\ep(t) + 1}^{\ep(t+k)} \frac{m^j_g}{m^j_b} \cdot   \quplus_{m^j_b}.$$
\end{lemma}
\begin{proof}
By the definition of \goodauction ,  bidder $i \in G_t$ is moved to bad state at the end of step $t$ if it bids lower than $r^g$. Since bidder $i$ bids lower than the reserve price, its utility in this round is $0$.

 Therefore, under this strategy, irrespective of the bids used in the rounds after $t$, the buyer $i$'s $k$-lookahead utility in round $t$ is at most the bad state utility over $\ep(t+k) - \ep(t) - 1$ epochs, which by Lemma~\ref{lem:ubutility}, is at most  
$$U_i^{[t, t+k]}(H_{t-1}, v_{i,t}, {\bf s}) \leq \sum_{j=\ep(t) + 1}^{\ep(t+k)} (1+\epsilon) \frac{\rho \epl^j}{m^j_b}  \quplus_{m^j_b} \leq \frac{2(1+\epsilon)\rho \threshold }{(1-\rho)(1-\delta)} \sum_{j=\ep(t) + 1}^{\ep(t+k)} \frac{m^j_g}{m^j_b} \cdot   \quplus_{m^j_b}.$$
 \end{proof}
 
 Next, we lower bound the $k$-lookahead utility achieved with the good strategy $s^g$ in the same situation.

\begin{lemma}
\label{lem:lbutilitylookahead}
Consider a round $t$ during \goodauction \    where the number of previous uncleared auctions is $U \geq nH / (1 - \delta)$. Then, for any buyer $i \in G$ playing  strategy $s_i = s^g$, $$U_i^{[t, t+k]}(H_{t-1}, v_{i,t}, {\bf s})  \geq (1-\epsilon) \threshold \left(-\quplus_{m^{\ep(t)}_g} +  \epsilon \sum_{j = \ep(t) + 1}^{\ep(t+k) - 1}\quplus_{m^j_g}\right).$$
\end{lemma}
\begin{proof}
The utility from bad state auctions is always non-negative, so we ignore that for the lower bound.  Now, since any buyer can be allocated at most $\threshold$ items in an epoch, the maximum payment in the remaining rounds of this epoch is at most $\resg \threshold$, lower bounding the utility by $-\resg \threshold=-\quplus_{m^\ep_g}(1-\epsilon)\threshold$. Now, consider the utility in any  of the next $\ep(t+k)-\ep(t) -2$ epochs. The buyer always remains in  good state or rest state in each of the next $\ep(t+k)-\ep(t)-2$ epochs. 
Combining the above argument with Lemma~\ref{lem:lbutility}, the $k$-lookahead utility of the new strategy is at least:
	\begin{eqnarray}
	\label{eq:2}
	& & \underbrace{-\quplus_{m^{\ep(t)}_g}(1-\epsilon)\cdot \threshold}_{\text{minimum utility from the rest of this epoch}} +\  \underbrace{ \epsilon(1-\epsilon)\threshold  \sum_{j = \ep(t) + 1}^{\ep(t+k) - 1}\quplus_{m^j_g}  }_{\text{minimum utility over next $\ep(t+k)-\ep(t)-2$ epochs}} \nonumber
	\end{eqnarray}
\end{proof}

The third and last lemma needed for the proof of Lemma~\ref{lem:dominating} is used to combine the two previous lemmas and argue that \goodauction \ lead  higher utility for a buyer  then \badauction. This lemma will also be used to argue that \goodauction \ give higher revenue.

\begin{lemma}
\label{lem:mgmb} For any $m_g, m_b$ such that $1 \leq m_g \leq n$ and $n/2 \leq m_b \leq n$, then 
$$\frac{1}{m_g} \quplus_{m_g} \geq   \cdot \frac{1}{n} \quplus_{n/2} \geq \frac{1}{2} \cdot  \frac{1}{m_b} \quplus_{m_b}.$$ 
\end{lemma}
\begin{proof}
First, for any $1 \leq m_1 \leq m_2 \leq n$, we have
$$\frac{1}{m_1} \quplus_{m_1} \geq \frac{1}{m_2} \quplus_{m_2}.$$
We obtain
\begin{align*}
\frac{1}{m_g} \quplus_{m_g} \geq \frac{1}{\max(n/2, m_g)} \quplus_{\max(n/2, m_g)} \geq \frac{1}{n} \quplus_{n/2} \geq \frac{1}{2}  \cdot \frac{1}{n/2} \quplus_{n/2} \geq \frac{1}{2}  \cdot  \frac{1}{m_b} \quplus_{m_b}
\end{align*} 
where the first inequality is since $m_g \leq \max(n/2, m_g)$, the second is since $\quplus_i$ is increasing in $i$, and the last since $n/2 \leq m_b$.
\end{proof}

We now prove that entering a bad state is a dominated strategy for lookahead buyers by comparing the bounds obtained by Lemma~\ref{lem:ubutilitylookahead} and Lemma~\ref{lem:lbutilitylookahead}.

\begin{proof}[Proof of Lemma~\ref{lem:dominating}]  With $k \geq  \Kbound  \eplmax$, we have $\sum_{j  = \ep(t) + 1}^{\ep(t+k) -1} \epl^j \geq \KboundE\eplmax$.
	Comparing the $k$-lookahead utility bounds obtained in Lemma~\ref{lem:ubutilitylookahead} and Lemma~\ref{lem:lbutilitylookahead}, the  strategy $s^g$ is dominating if 
	
$$(1-\epsilon) \threshold \left(-\quplus_{m^{\ep(t)}_g} +  \epsilon \sum_{j = \ep(t) + 1}^{\ep(t+k)-1}\quplus_{m^j_g}\right) \geq  \frac{2(1+\epsilon)\rho \threshold }{(1-\rho)(1-\delta)} \sum_{j=\ep(t) + 1}^{\ep(t+k)} \frac{m^j_g}{m^j_b} \cdot   \quplus_{m^j_b}$$

We show the following three inequalities that when combined give the above inequality:
\begin{eqnarray}
	\label{eq:in1}
	\frac{1}{3}(1-\epsilon)  \epsilon \sum_{j = \ep(t) + 1}^{\ep(t+k)-1}\quplus_{m^j_g} & \geq &\frac{2(1+\epsilon)\rho  }{(1-\rho)(1-\delta)} \sum_{j=\ep(t)+ 1}^{\ep(t+k)-1} \frac{m^j_g}{m^j_b} \cdot   \quplus_{m^j_b} \\
	\label{eq:in2}
	\frac{1}{3}(1-\epsilon) \epsilon \sum_{j = \ep(t)+ 1}^{\ep(t+k)-1}\quplus_{m^j_g} & \geq & (1-\epsilon)  \quplus_{m^{\ep(t)}_g} \\
	\label{eq:in3}
\frac{1}{3}(1-\epsilon)  \epsilon \sum_{j = \ep(t)+ 1}^{\ep(t+k)-1}\quplus_{m^j_g} & \geq &\frac{2\rho }{(1-\rho)(1-\delta)} \frac{m^{\ep(t+k)}_g}{m^{\ep(t+k)}_b} \cdot   \quplus_{m^{\ep(t+k)}_b} 
\end{eqnarray}

We first show inequality~(\ref{eq:in1}). By Lemma~\ref{lem:mgmb}, we have $\frac{1}{m^j_g} \quplus_{m^j_g}\geq \frac{1}{2} \cdot  \frac{1}{m^j_b} \quplus_{m^j_b}$ for all $j \in [\ep(t)+ 1, \ep(t+k)-1]$. Inequality~\ref{eq:in1} then holds by the assumption that $\rho \le \rhobound$.

For inequalities~(\ref{eq:in2}) and~(\ref{eq:in3}), we first have
\begin{align*}
 \sum_{j = \ep(t)+ 1}^{\ep(t+k)-1}\quplus_{m^j_g} & \geq  \frac{1}{n} \quplus_{n/2}  \sum_{j = \ep(t)+ 1}^{\ep(t+k)-1}m^j_g  & \text{Lemma~\ref{lem:mgmb}}\\
& \geq   \frac{1}{n} \quplus_{n/2}  \frac{(1 - \rho)(1 - \delta)}{2\threshold} \sum_{j = \ep(t)+ 1}^{\ep(t+k)-1} \epl^j & \text{Definition of }  \epl^j \\
& \geq   \frac{1}{n} \quplus_{n/2}  \frac{(1 - \rho)(1 - \delta)}{2\threshold} \KboundE \eplmax & \text{Lower bound on } k\\
& \geq    \frac{8 \quplus_{n/2} }{\epsilon(1- \epsilon)} & \text{Definition of }  \eplmax\\
& \geq    \frac{4 \quplus_{n} }{\epsilon(1- \epsilon)} &  \quplus_{n/2}  \geq \quplus_{n} /2
\end{align*}
Inequality~(\ref{eq:in2}) then holds since  $\quplus_{n} \geq \quplus_{m^{\ep(t)}_g} $. For inequality~(\ref{eq:in3}), note that  $$\quplus_{n} \geq \quplus_{m^{\ep(t+k)}_g} \geq \frac{1}{2} \frac{m^{\ep(t+k)}_g}{m^{\ep(t+k)}_b} \cdot   \quplus_{m^{\ep(t+k)}_b} $$ where the second inequality is by Lemma~\ref{lem:mgmb} and  that $\rho \le \rhobound$.
\end{proof}

\begin{lemma}	
  \label{lem:myopic}
Any  myopic buyer $i$ playing an undominated strategy bids $b_{i,t} \ge \resb$ if its valuation is $v_{i,t} > \resb$ in any round $t$ of \badauction where $i$ is in bad state.
\end{lemma}
\begin{proof}
Let $B$ be the set of bad state buyers in a round $t$ of \badauction of epoch $\ep$. We argue that for every myopic buyer $i \in B$ that has private valuation $v_{i,t} > \resb$, bidding below $\resb$ is a dominated strategy.  A myopic buyer maximizes its utility in the current round $t$. With bid $b_{i,t} < \resb$, buyer $i$ gets zero utility in round $t$. If $v_{i,t} > \resb$,  bidding $b_{i,t} < \resb$ is dominated by the strategy of bidding  $b_{i,t} = \resb$, which always obtains non-negative utility in round $t$ and obtains strictly positive utility when other buyers bid below $\resb$.
\end{proof}
 



\section{Strategies of No-regret Buyers with Heterogeneous Learning Behavior}
\label{sec:learning}




\begin{lemma}
\label{lem:policy}
Every no-policy regret learner must remain in good state in all but $o(T)$ rounds.
\end{lemma}
\begin{proof}
Here, we use as benchmark expert strategy $s^g \in E$ of bidding $\resg=(1-\epsilon)\quplus_{m_g}$ in good state auctions whenever $v_{i,t}\ge \qu_{m_g}$ initially, and bidding $\resg$ continuously once uncleared auctions reach the limit ($m_g \threshold/(1-\delta)$). 

This bidding strategy ensures that the buyer is always in a good state, and achieves an expected  utility of at least $\epsilon(1- \epsilon)\quplus_{m^i_g}\threshold$ during epoch $i$ by Lemma~\ref{lem:lbutility}. Since $\epl = m^i_g\threshold(1+\frac{1}{1-\delta})\frac{1}{(1-\rho)} \le \frac{1}{(1-\rho)(1-\delta)}2 m_g\threshold$, the expected per round utility is at least
$$  \frac{(1-\rho)(1-\delta)\epsilon(1- \epsilon)\quplus_{m^i_g}}{2 m^i_g}$$
during that epoch, for some $m^i_g \in [n]$.  Therefore, since the class $E$ contains such sequences of single experts, the policy-regret learning buyer must achieve at least $\frac{(1-\rho)(1-\delta)\epsilon(1- \epsilon)\quplus_{m^i_g}}{2 m^i_g} - o(1)$ utility per round. Now, once in bad state, the buyer can achieve at most $$\frac{\rho  \quplus_{m^i_b}}{m^i_b} $$ utility on average by Lemma~\ref{lem:ubutility} for some $m^i_b$ such that $n/2 \leq m_i^b \leq n$. By Lemma~\ref{lem:mgmb}, for any $m_g, m_b$ such that $1 \leq m_g \leq n$ and $n/2 \leq m_b \leq n$, then 
$\frac{1}{m_g} \quplus_{m_g} \geq \frac{1}{2} \cdot  \frac{1}{m_b} \quplus_{m_b}.$
 
 Therefore, if $\rho < (1-\rho)(1-\delta)\epsilon(1- \epsilon)/4 - \Omega(1)$, then the number of bad state epochs in buyer's state trajectory can be at most $o(T)$. This implies that the learner remains in a good state for at least $T-o(T)$ rounds. 
\end{proof}

We use the set of experts $E$ as defined in Definition~\ref{def:E} in the lemmas below.

\begin{lemma}
\label{lem:noregret}
Consider $\Gamma_i$ as the set of rounds $t$ where buyer $i$ is participates in a bad state auction, $v_{i,t} > \myeprice$ and $b_{j,t}<\resb$ for all other buyers $j$ in bad state. If buyer $i$ is a no-regret learner against the set of experts $E$ in Definition \ref{def:E}, then the buyer must bid $b_{i,t}\ge \resb$ for all but $o(T)$ of rounds in $\Gamma_i$.
\end{lemma}
\begin{proof}
Consider a bidding expert function $f(h, v)$ defined as $f(h, v)=\resb$ when the projected history $h$ indicates that the buyer is in the bad state and valuation $v\ge \myeprice$; and $0$ otherwise. 
This is (or is arbitrarily close in case of discretization) one of the experts in the class $E$ of experts that the buyer is using.  Consider any round $t$ where $v_{i,t}\ge \myeprice$. On bidding reserve price in this round, if bids  $b_{j,t}<\resb$ for all the other buyers, then this buyer wins the auction. For a given trajectory of valuations, and other buyers' bids, 
$\Gamma_i$ denotes the set of such rounds in \badauction; specifically, 
$\Gamma_i=\{t\in [T]: t\in \epl^b \text{ for some epoch}, i\in B \text{ (bad state)}, v_{i,t}\ge \myeprice, b_{j,t}<\resb, j\ne i\}$.
 Therefore, the hindsight utility of this expert is at least 
$$\textstyle \sum_{t=1}^T u_{i,t}(f(h_{i,t-1},v_{i,t})) \ge   \sum_{t\in \Gamma_i} (v_{i,t} -\resb) \ge |\Gamma_i| \frac{\epsilon}{n} \qu_n = \Omega(|\Gamma_i|)$$

Since the buyer $i$ is using a no-regret learning algorithm, she must be achieving a utility that is within $o(T)$ of the above utility for every trajectory. Now, the buyer cannot make any positive utility in a bad state in round $t$ if $b_{i,t} \le \resb$. Therefore, on any given trajectory, a no-regret learning buyer must have bid $b_{i,t}\ge \resb$ in all but $o(T)$ of rounds in $\Gamma_i$. 
\end{proof}
A corollary of the above lemma is that if all buyers are learning buyers at most $o(T)$ of bad state auctions where some participating no-regret buyer has valuation above $\myeprice$ can go uncleared.
\begin{corollary}
Consider the set of rounds among bad state auctions where all buyers in bad state are no-regret learners against experts $E$, and at least one buyer in bad state has valuation of $\myeprice$ or more. Then, in all but $o(T)$ such rounds, at least one buyer $i$ with $v_{i,t}\ge \myeprice$ is guaranteed to bid $b_{i,t}\ge \resb$.
\end{corollary}

\section{Proof of the Main Result}
\label{sec:mainresult}
We prove the main result, which is a bound on the revenue obtained by the mechanism under a heterogeneous buyer population consisting of an unknown proportion of $k$-lookahead buyers for different $k$,  myopic buyers, no-policy regret learners, and no-regret learners 

\begin{theorem}
Let $\nsoph$ be the number buyers which are either $k$-lookahead,  for $k \ge \Kbound  \eplmax$, or no-policy regret learners. Let $\nnaive$ be the number of remaining   buyers which are either myopic or no-regret learners. Assume Algorithm~\ref{alg:one} is instantiated with parameters $\rho,\epsilon, \delta \in (0,1)$ s.t. $\rho \le \rhobound$. If  the myopic and lookahead buyers play undominated strategies and the learners play no-regret or no-policy regret strategies, then the expected per round revenue is at least
   $$(1-\epsilon)(1-\delta)(1-\rho)\frac{1}{4}\cdot  \quplus_{\nsoph}  + \frac{\rho(1-\epsilon)}{2} \left(1-\frac{1}{e}\right)\Rev^{\Mye}(\nnaive)  - o(1)$$
 Here, $\quanplus=\Ex_{v\sim F}[v|v\ge \quan]$, with $\quan=F^{-1}(1-1/n)$ being the $n^{th}$ quantile for the valuation distribution, and $\quplus_0 = 0$.  $\Rev^{\Mye}(n)$ is the optimal revenue in a single-item auction with $n$ buyers.
 \end{theorem}
 The following corollary can be obtained by simple algebraic manipulations of the result stated in the above theorem. 
 
\thmmain*

\begin{proof}[Proof of Theorem~\ref{thm:main}] Let $\nsoph$ be the  number of buyers which are either $k$-lookahead with $k \ge \Kbound  \eplmax$ or no-policy regret learners and $\nnaive$ be the remaining buyers (myopic or no-regret learner). We define $\badrounds$ to be the collection of rounds where there is either a no-policy regret learner that is in bad state, a no-regret learner that bids $b_{i} < \resb$ when $v_i > \myeprice$.

 Consider an epoch $\ep$ where none of the rounds during that epoch are  in $\badrounds$ and where no buyer is moved from good state to bad state. Let $|G|$ and $|B|$ denote the number of good state and bad state buyers in the beginning of this epoch.
 
 By  Lemma \ref{lem:dominating}, all $k$-lookahead buyers (for $k\ge \Kbound  \epl$) remain in good or rest state in all epochs.  Since none of the rounds during epoch $\ep$ are in $\badrounds$, every no-policy regret learner remains in good state in epoch $\ep$ by definition of $\badrounds$.  Therefore, $|G| = \nsoph + \nnaiveg$ and $|B| = \nnaive - \nnaiveg$ for some $\nnaiveg \geq 0$. By 
Lemma \ref{lem:revGood}, we get that  the expected total revenue from  \goodauction  \ in this epoch is at least $|G| \threshold (1 - \epsilon) \quplus_{m_g}$.
 
 For \badauction, we first recall that there is no-policy regret learner in bad state at epoch $\ep$. By Lemma~\ref{lem:myopic}, every myopic  buyer bids at least $\resb$ if $v_{i} \geq \myeprice$ during a round of \badauction. By definition of $\badrounds$, ever no-regret learner bids at least $\resb$ if $v_{i} \geq \myeprice$ during a round of \badauction \ at epoch $\ep$. by Lemma~\ref{lem:revenuebad}, the expected revenue of the mechanism from \badauction \ at epoch $\ep$ is thus at least $(1-\epsilon)(1-1/e) \frac{|B|}{m_b} \rho\epl \revmye(m_b)$.

The expected total revenue from this epoch is thus at least
$$|G| \threshold (1 - \epsilon) \quplus_{m_g} + |B|(1 - \epsilon) \left(1-\frac{1}{e}\right) \rho \epl \frac{\Rev^{\Mye}(m_b)}{m_b}$$

To bound this expected revenue,  we consider two cases.

\paragraph{If $|B| \geq n/2$.} In this case, $m_b = |B|$. Since $m_g, |G| \leq n/2 \leq |B| = \nnaive$, we get

$$\quplus_{m_g} \geq \frac{m_g}{\nnaive} \cdot \quplus_{\nnaive} \geq \frac{m_g}{\nnaive}\revmye(\nnaive).$$

Since $\epl = \eplVal$ and $\rho = \rhobound$, this implies that
$$ \threshold (1 - \epsilon) \quplus_{m_g}  \geq \frac{(1- \epsilon)(1 - \delta)(1- \rho)\epl}{2} \cdot \frac{\revmye(\nnaive)}{\nnaive} \geq (1 - \epsilon) \left(1-\frac{1}{e}\right) \rho \epl \frac{\revmye(\nnaive)}{\nnaive}.$$ 

Next, since $|B| = m_b$ and $\nnaive \geq |B|$, we have $$|B|  \frac{\Rev^{\Mye}(m_b)}{m_b} = \Rev^{\Mye}(|B|) \geq \frac{|B|}{\nnaive}\Rev^{\Mye}(\nnaive).$$
Since $ |G|  \geq |G|/2+ \nnaiveg/2$, we conclude that the total expected revenue from this epoch in this case is at least 
\begin{align*}
& \frac{|G|}{2}(1 - \epsilon)  \threshold \quplus_{m_g}  + \left(\frac{\nnaiveg}{2} + |B|\right)(1- \epsilon)(1-\frac{1}{e})  \frac{\rho \epl}{\nnaive} \Rev^{\Mye}(\nnaive)  \\
\geq & \frac{|G|}{2}(1 - \epsilon) \threshold  \quplus_{m_g}  + \frac{1}{2} (1- \epsilon)(1-\frac{1}{e}) \rho \epl \Rev^{\Mye}(\nnaive)
\end{align*}

where the inequality is since $\nnaive = \nnaiveg + |B| \leq 2(\nnaiveg/2 + |B|)$.

\paragraph{If $|B| < n/2$.} In this case, we have $m_g = |G| \geq n/2 \geq \nnaive / 2$. Wet get
$$\quplus_{m_g} \geq \quplus_{\nnaive / 2} \geq \frac{1}{2} \quplus_{\nnaive} \geq \frac{1}{2} \revmye(\nnaive).$$
Since $\epl = \eplVal$ and $\rho = \rhobound$, this implies that
$$ \threshold (1 - \epsilon) \quplus_{m_g}  \geq \frac{1}{4}(1- \epsilon)(1 - \delta)(1- \rho)\epl\cdot \revmye(\nnaive)\geq (1 - \epsilon) \left(1-\frac{1}{e}\right) \rho \epl \revmye(\nnaive).$$ 
 Thus, the expected total revenue from this epoch in this case is at least 
\begin{align*}
|G|(1 - \epsilon)\threshold  \quplus_{m_g}  & = \frac{1}{2}(1 - \epsilon)|G|\threshold  \quplus_{m_g}  + \frac{1}{2} (1 - \epsilon)|G|\threshold  \quplus_{m_g} \\
& \geq \frac{1}{2} (1 - \epsilon)|G|\threshold  \quplus_{m_g} + \frac{1}{2}(1 - \epsilon) \left(1-\frac{1}{e}\right) \rho \epl \revmye(\nnaive)
\end{align*}

Since  $|G| \geq \nsoph$, we have $\quplus_{m_g} \geq \quplus_{\nsoph}$. 
Since $\epl = \eplVal$, in both cases, the expected per round revenue of  epoch $\ep$ is thus at least 
  $$\frac{(1 - \epsilon)(1-\delta)(1-\rho)}{4} \quplus_{\nsoph} + \frac{1}{2}(1 - \epsilon) \left(1-\frac{1}{e}\right) \rho  \revmye(\nnaive)$$
  where $\quplus_{\nsoph} = 0$ if $\nsoph = 0$
  
   We considered an epoch $\ep$ where none of the rounds during that epoch are  in $\badrounds$ and where no buyer was moved from good state to bad state. First, since once a buyer is moved to bad state they remain in bad state for every future epoch, there are at most $n$ epochs where a  buyer is moved from good state to bad state. We discount $\eplmax \cdot \max \resg$  revenue for each of those $n$ epochs.  Second, by Lemma~\ref{lem:policy} and \ref{lem:noregret}, $|R| = o(T)$. Thus, the number of epochs where there is at least one round during that epoch that is in $\badrounds$ is $o(T)$. We discount $\eplmax  \cdot \max \resg$  revenue for each of those $o(T)$ epochs. We obtain the lower bound on per-round revenue as stated in the theorem statement.
  \end{proof}

%

\section{Regarding Existence of Policy Regret Learning Algorithm} \label{sec:nprexistence}
In general, achieving $o(T)$ policy regret is difficult; \cite{ADT12} show that there exists an adaptive adversary such that any learning algorithm has regret at least $\Omega(T)$. However, it sufficient for us to consider policy regret learners against a small set of experts $E$  (e.g., $O(1)$ size) and further when other buyers  are restricted to not use history beyond $o(T)$ past steps. Under such restrictions a simple learning algorithm is to initially explores every expert for $o(T)$ steps and then use the expert with best performance for remaining time steps. We can show that this simple learning algorithm achieves $o(T)$ policy regret under our mechanism with a small modification of resetting the mechanism {\it once} after $o(T)$ steps so that this initial exploration does not hurt buyer's utility in future rounds. This modification does not affect any of the revenue guarantees provided by our mechanism. Further discussion is provided in Appendix~\ref{sec:mainresult}.

\end{document}